\documentclass{article}
\usepackage{amsmath,amsthm,amssymb}
\usepackage{amssymb,amsmath,amsthm,graphicx,epsfig,latexsym}

\newcommand\beq{\begin{equation}}
\newcommand\eeq{\end{equation}}
\newcommand\bea{\begin{eqnarray}}
\newcommand\eea {\end{eqnarray}}

\newcommand\ba{\begin{array}}
\newcommand\ea {\end{array}}

\newtheorem{Theorem 1}{Theorem}
\newtheorem{Theorem 2}[Theorem 1]{Theorem}

\def\O{{\Omega}}

\def\a{{\alpha}}

\def\g{{\gamma}}
\def\p{{\phi}}

\def\e{{\epsilon}}

\def\la{{\Lambda}}

\def\P{\mathbf{P}}

\def\a{\alpha}

\def\g{\gamma}

\def\un\a{{\underline\alpha}}

\def\pe{\phi_\epsilon}
\def\PF{\Phi^{F*}}

\font\german=eufm10 at 10pt
\def\Buchstabe#1{{\hbox{\german #1}}}

\def\CP{\mathbb{CP}}
\def\R{\mathbb{R}}
\def\ZZ{\mathbb{Z}}
\def\EA{\Buchstabe{A}}      
\def\CE{\Buchstabe{A}^*}    
\def\Z2{\mathbb{Z}_2}
\def\CE{{\EA^*}}        
\def\H{(\Omega,\EA,\mu)}     
\def\C{{\cal{C}}(\Omega,\EA,\mu)} 
\def\M{{\cal{M}}(\Omega,\EA,\mu)} 
\def\MF{{\cal{M}}_F^*}            
\def\Mstar{{\cal{M}}^*}           
\def\Mint{{\cal{M}}^*_\cap}       

\newtheorem{theorem}{Theorem}

\newtheorem{lemma}{Lemma}

\newtheorem{definition}{Definition}

\theoremstyle{remark}

\begin{document}

\begin{center}
   \baselineskip=16pt
   \begin{LARGE}
      \textsl{Dynamics \& Predictions in the Co-Event Interpretation}
   \end{LARGE}
   \vskip 1cm
      Yousef Ghazi-Tabatabai
   \vskip .2cm
   \begin{small}
      \textit{Blackett Laboratory, Imperial College,\\
        London, SW7 2AZ, U.K.}
    \end{small}
   \vskip 0.5cm
      Petros Wallden
   \vskip .2cm
   \begin{small}
      \textit{Raman Research Institute \\
              Bangalore - 560 080, India}
   \end{small}
   \vskip 1cm
\end{center}


\begin{abstract}

\noindent Sorkin has introduced a new, observer independent, interpretation of quantum mechanics that can give a successful realist account of the `quantum micro-world' as well as explaining how classicality emerges at the level of observable events for a range of systems including single time `Copenhagen measurements'. This `co-event interpretation' presents us with a new ontology, in which a single `co-event' is real. A new ontology necessitates a review of the dynamical \& predictive mechanism of a theory, and in this paper we begin the process by exploring means of expressing the dynamical and predictive content of histories theories in terms of co-events.

\end{abstract}

\section{Introduction}

\subsection{Opening Remarks}

Building on the work of Dirac \cite{Dirac:1933} and Feynman \cite{Feynman:1948,Feynman:1965}, the \emph{histories approach} to quantum mechanics \cite{Griffiths:1984rx,Omnes:1988ek,Gell:1990,Hartle:1992as,Sorkin:1994dt} is a reformulation and generalisation of (Copenhagen) quantum mechanics in an observer independent framework.  With a `space-time history' replacing the `state' as the fundamental object of the theory, the dynamics can be rephrased as a `decoherence functional' (or alternatively a `quantum measure') on subsets of the `sample space' of histories. Observer independence is then to be sought through the `embedding' of the (Copenhagen) concept of experimentally observable events in the wider (though ill defined) notion of `classical partitions'; which are typically associated with the property of decoherence (dynamical classicality).

However the need for an interpretation remains. Because the mathematical formalism of a histories theory is reminiscent of probability theory, a naive application of the (relatively) well understood `one history is real' classical interpretation may seem appropriate; however it fails to explain gedankenexperimentally realisable theories based on the Kochen-Specker theorem \cite{Dowker:2007zz}.

The \emph{consistent histories} interpretation \cite{Griffiths:1996,Griffiths:1998} is perhaps the most developed alternative. Identifying as `classical' every dynamically classical (decoherent) partition, it allows us to simultaneously assign truth valuations to all propositions within a single \emph{consistent set} (a classical partition), within the `context' of that consistent set, in a manner satisfying Boolean logic. However, two propositions that do not participate in a common consistent set can not be simultaneously (truth) evaluated.

Sorkin has proposed a new interpretation \cite{Sorkin:2006wq,Sorkin:2007,Dowker:2007zz} that can assign truth valuations to all elements of the event algebra simultaneously. This is achieved by moving from the classical `one history is real' to an interpretation based on `one co-event is real'; where a co-event is associated with a truth valuation map $\p:\EA\rightarrow\Z2$ obeying certain restrictions.

A new ontology necessitates a review of the dynamical \& predictive mechanism of a theory. How do our observations and predictions relate to the objects we have identified as (potentially) real? Can this innovation help us to develop a better understanding of `classicality', or a more complete theory of measurement? In this paper we make a start at answering such questions by exploring means of expressing the dynamical and predictive content of histories theories in terms of co-events.

\subsection{Outline of this Paper}

In section \ref{sec:quantum measure theory} we introduce quantum measure theory, a formulation of the histories approach which phrases the dynamics in a fashion resembling a probability measure. In section \ref{sec:the coevent interpretation} we introduce the co-event interpretation, reviewing its application to classical stochastic theories before turning to the multiplicative scheme. We are then in a position to set out and motivate the goals of this paper in section \ref{sec:goals}. In section \ref{sec:dynamics on coevents} we discuss the possibility of `moving wholesale' from histories theories to `co-event theories', in which the dynamics is described by a probability measure on the space of potentially real co-events. Unfortunately this approach stalls, so in section \ref{sec:approximate preclsuion} we focus our attention to experimentally falsifiable predictions, and via Cournot's Principle (section \ref{sec:cournot's principle}) attempt to explain these through \emph{approximate preclusion} (section \ref{sec:approximate co-events}); which we apply in a manner consistent with Literal Strong Cournot (section \ref{sec:strong cournot}), and then in a manner consistent with Operational Weak Cournot (section \ref{sec:quantum weak cournot}). We conclude in section \ref{sec:conclusion}; appendix \ref{appendix:principle classical partition} reviews some results regarding the principle classical partition.

\subsection{Quantum Measure Theory}\label{sec:quantum measure theory}

\subsubsection{The Quantum Measure}

Quantum measure theory \cite{Sorkin:1994dt} is a histories based attempt to approach quantum mechanics as a generalisation of classical stochastic theories; phrasing the quantum dynamics as a generalisation of the probability measure. We start with a brief review of quantum measure theory and refer to \cite{Sorkin:1994dt,Sorkin:1995nj,Salgado:1999pu,
Sorkin:2006wq,Sorkin:2007} for more details.

A \emph{histories theory} (sometimes referred to as a \emph{generalised measure theory}) consists of a triple, $(\O, \EA, \mu)$, of a `sample space' of histories, an event algebra and a measure. The sample space, $\O$, contains all the ``fine grained histories'' or ``formal trajectories'' for the system  {\it e.g.} for $n$-particle mechanics -- classical or quantum -- a history would be a set of $n$ trajectories in spacetime, and for a scalar field theory, a history would be a field configuration on spacetime. In a classical theory $\O$ is the usual measure theoretic sample space.

The event algebra, $\EA$, contains all the (unasserted) propositions that can be made about the system. We will call elements of $\EA$ {\textit{ events}}, following standard terminology in the theory of stochastic processes. Now the power set, $P\O$, of $\O$ is usually thought of as a Boolean algebra under union and intersection, however for our purposes it will be more useful to view $P\O$ as a ring (or an algebra over $\Z2$), with symmetric difference playing the role of addition ($A+B := (A\cup B)\setminus (A\cap B)$) and intersection playing the role of multiplication ($AB:=A\cap B$). In cases where $\O$ is finite, $\EA$ can be identified with the whole power set, $\EA=P\O$, when $\O$ is infinite, $\EA$ can be identified with an sub-ring of the power set, $\EA \subset P\O$. In both cases, the algebraic properties of $\EA$ play a central role in the formulation of the co-event interpretation. In a classical theory $\EA$ will typically be the $\sigma$-algebra (or $\delta$-algebra) of measurable sets.

Predictions about the system --- the dynamical content of the theory ---  are to be gleaned, in some way or another, from
a generalized measure $\mu$, a non-negative real function on $\EA$. $\mu$ is the dynamical law and initial condition rolled into one. In a classical theory, $\mu$ will be a probability measure. In a quantum theory, $\mu$ will obey \cite{Sorkin:1994dt,Salgado:1999pu}:
\begin{enumerate}
\item \textit{Positivity} \beq\label{eq:quantum measure positivity} \mu(A)\geq 0, \eeq
\item \textit{Sum Rule} \beq\label{eq:quantum measure sum rule} \mu(A\sqcup B\sqcup C) = \mu(A\sqcup B) + \mu(B\sqcup C) + \mu(A\sqcup C) -\mu(A) - \mu(B) -\mu(C), \eeq
\item \textit{Unitality} \beq\label{eq:quantum measure unitality} \mu(\O)=1, \eeq
\end{enumerate}
where $A,~B,~C\in\EA$ are disjoint elements of $\EA$, as indicated by the symbol `$\sqcup$' for disjoint union. Given a histories theory, we can construct the following series of symmetric set functions,which are sometimes referred to as the Sorkin hierarchy\footnote{These are the generalised interference terms introduced in \cite{Sorkin:1994dt}}:
\begin{align*}
        I_1(X) & \equiv \mu(X), \\
      I_2(X,Y) & \equiv \mu(X\sqcup Y) - \mu(X) - \mu(Y), \\
  I_3(X, Y, Z) & \equiv \mu(X\sqcup Y \sqcup Z) -
      \mu(X\sqcup Y) - \mu(Y\sqcup Z) - \mu(Z\sqcup X) \\
{}&\ + \mu(X) + \mu(Y) + \mu(Z),
\end{align*}
and so on, where $X$, $Y$, $Z$, {\it etc.\ }are disjoint. A histories theory is \emph{of level $k$} (alternatively it is a {\it measure theory of level $k$}) if the sum rule $I_{k+1}=0$ is satisfied. It is known that this condition implies that all higher sum rules are automatically satisfied, namely $I_{k+n}=0$ for all $n\geq 1$ \cite{Sorkin:1994dt}.

A level $1$ theory is thus one in which the measure satisfies the usual Kolmogorov sum rules of classical probability theory, $\P(A\sqcup B) = \P(A) + \P(B)$, classical Brownian motion being a good example. We refer to level one theories as \emph{classical theories}.  A level $2$ theory is one in which the Kolmogorov sum rules may be violated but $I_3$ is nevertheless zero. Any unitary quantum theory can be cast into the form of a generalised measure theory and its measure satisfies the condition $I_3 = 0$, which trivially implies equation \ref{eq:quantum measure sum rule}. We refer to level 2 theories, therefore, as \emph{quantum measure theories}. It is known that whenever a system can be described by standard (Copenhagen) quantum mechanics via the usual Hilbert space construction (through what we will call a \emph{Hilbert space theory}), we can also find a quantum measure theory describing the system. We say that the histories theory is derived from the Hilbert space theory; the quantum measure defined using Feynman path integrals \cite{Isham:1994uv}.

Notice that the replacement of the Kolmogorov sum rule with equation \ref{eq:quantum measure sum rule} has significant consequences for the null set structure (a null set is a set of measure zero); in particular because of (destructive) interference negligible sets (subsets of null sets) are no longer null in general. Unfortunately we lack general results about the structure of null sets, which is a hinderance to the development of our interpretations of quantum measure theory. There is recent research activity in this area, for example work on the implications of the existence of `antichains' of null sets \cite{Surya:2008}.

In what follows, unless explicitly stated otherwise, we shall assume all histories theories to be level $2$. Further, because of the current uncertainty as to the correct general construction of $\mu$ and $\EA$ in the case of an infinite sample space, we shall unless specifically mentioned otherwise, assume that all sample spaces are finite and that $\EA=P\O$.

The existence of a quantum measure, $\mu$, is more or less equivalent \cite{Sorkin:1994dt} to the existence of a
{\it decoherence functional}, $D(\,\cdot\,\,,\,\cdot\,)$, a complex function on $\EA\times \EA$ satisfying \cite{Hartle:1989,Hartle:1992as}:

\noindent (i) Hermiticity: $D(X\,,Y) = D(Y\,,X)^*$ ,\  $\forall X, Y \in \EA$;

\noindent (ii) Additivity: $D(X\sqcup Y\,, Z) = D(X\,,Z) + D(Y\,,Z)$ ,\
   $\forall X, Y, Z \in \EA$ with $X$ and $Y$ disjoint;

\noindent (iii) Positivity: $D(X\,,X)\ge0$ ,\  $\forall X\in \EA$;

\noindent (iv) Normalization: $D(\Omega \,,\Omega)=1$.\footnote{The normalisation
condition may turn out not to be necessary, but we include it
because all the
quantum measures we consider in this paper will satisfy it.}

\noindent The quantum measure is related to the decoherence functional by
\begin{equation}
\mu(X) = D(X \,, X)\quad \forall X \in \EA \,.
\end{equation}
The quantity $D(X,Y)$ is interpretable as the quantum interference between two sets of histories in the case when $X$ and $Y$ are disjoint.

Finally, we end this section by introducing the important concept of \emph{coarse graining} \cite{Hartle:1992as,Isham:1994uv}, which is crucial to the study of subsystems \& emergent classicality.

\begin{definition}
Let $\H$ be a histories theory with a finite sample space. We say that a histories theory $(\la,\EA_\la,\mu_\la)$ is a \textbf{coarse graining} of $\H$ if:
\begin{enumerate}
  \item $\la$ is a partition of $\O$
  \item $\EA_\la = P\la$
  \item $\mu_\la(A) = \mu(A)~~\forall~A\in\EA_\la$
\end{enumerate}
We refer to $\H$ as a \textbf{fine graining} of $(\la,\EA_\la,\mu_\la)$. In this context we sometimes refer to the histories $\g\in\O$ as \textbf{fine grained histories} to distinguish them from the \textbf{coarse grained} $A\in\O_\la$.
\end{definition}

\subsubsection{Classicality and Observation}\label{sec:dynamical classicality}

We need to relate this framework to our (experimental) observations, and to develop the ability to make falsifiable predictions. Now we envision that given a histories theory $\H$, some of the events in $\EA$ may be `observable'; so that we can determine their truth or falsity though our experience. However not all the elements of $\EA$ will in general be observable, leading us enquire into the properties the set of observable events.

In `Copenhagen quantum mechanics' described by a `Hilbert space theory' we assume that the `observables' are given; determined by an assumed measuring apparatus that is external to the system. Further, we can assign (Copenhagen) probabilities to each of these events; these probabilities constitute the predictions of the theory\footnote{Though as we shall see a `probability statement' such as $\P(A)=p$ is not in general falsifiable through a single trial; we must construct a methodology for testing such assertions by relating them to statements that are falsifiable.}. If a histories theory $\H$ is derived from such a Copenhagen framework, then we inherit a set of `observable events', each of which is an element of an `observable partition' (of $\O$) whose elements are all observable events. Further, it can be shown that such partitions are \emph{decoherent}; so that the decoherence functional $D$ associated with the measure obeys $D_\la(X,Y)=0$ whenever $X,Y\in\la$ and $X\neq Y$ \cite{Hartle:1992as}. It can be shown that this implies that the coarse grained measure $\mu_\la$ obeys the axioms of a probability measure, and that $\mu_\la(A)$  is equal to the `Copenhagen probability' of the event $A$. These `probability statements' regarding observable events form the predictive content of such theories.

However for a general quantum measure theory we do not have a fully developed framework for dealing with observable events. Further, one of the key themes of the histories approach is observer independence; observations are certainly not observer independent. This leads us to `embed' the notion of (experimentally) observable events into the wider (but also at present ill defined) notion of `classicality'; which we require to be observer independent. Thus we wish to distinguish some partitions as being `classical' in some objective and observer independent manner; these \emph{classical partitions} are to possess the properties we require of observable partitions, and are to be interpreted in the same manner. The idea is that the set of classical partitions will contain the set of observable partitions; thus the `classical' behaviour of observed events is to be thought of an instance of an objective and observer independent phenomenon.

However, as alluded to above we are as yet uncertain what this broader phenomenon of classicality should be. Perhaps the simplest suggestion is that of \emph{dynamical classicality}, or \emph{decoherence}. As we saw above, `Copenhagen' observable partitions are decoherent; further, when the measure is classical the whole sample space is decoherent. This leads us to the notion of dynamical classicality, which identifies classicality with decoherence; thus a partition $\la$ of $\O$ is dynamically classical if and only if $D(X,Y)=0$ for all $X,Y\in\la$ such that $X\neq Y$. As before, it can be shown that this implies that the coarse grained measure $\mu_\la$ obeys the axioms of a probability measure. The notion of \emph{approximate decoherence}, $D(X,Y)<\e<<1$ for all $X,Y\in\la$ such that $X\neq Y$, is sometimes offered as an alternative to decoherence. In particular, in the study of environmental decoherence, observed `emergent classical behaviour' is due to approximate decoherence \cite{GellMann:1992kh,Dowker:1994dd}. However for the purposes if this paper we will focus on the `full' decoherence condition. Thus, given a general quantum measure theory, we will assume that we can identify the observable events, that each observable event participates in an observable partition, and that all observable partitions are dynamically classical.

Now given a single dynamically classical partition $\la$, we can treat the coarse grained histories theory $(\la,\EA_\la,\mu_\la)$ as a classical theory and apply our usual interpretation; identifying a single history as real. We are then able to simultaneously assign truth values (in $\Z2$) to every element in $\la$. However, when the underlying theory $\H$ is not itself classical, we may have various \emph{incompatible} (or \emph{incomparable}) partitions \cite{Griffiths:1996}; we are not able to simultaneously assign truth values to all events in two incompatible partitions using the classical `one history is real' interpretation \cite{Dowker:1994dd}. The \emph{consistent histories} interpretation addresses this issue by interpreting every `classical proposition' (an event that is an element of a dynamically classical partition) as \emph{conditional} on the (`largest') dynamically classical partition in which it participates. Essentially, it is only meaningful to assign truth valuations to classical partitions, and this must be done within the context of a given dynamically classical (or \emph{consistent}) partition; two propositions can only be simultaneously answered if they both participate in a common consistent partition \cite{Dowker:1994dd,Griffiths:1996,Griffiths:1998}.

Sorkin \cite{Sorkin:2006wq,Sorkin:2007} has introduced a new interpretation of quantum measure theory that can assign truth valuations to all elements of the event algebra simultaneously, rather than to individual consistent partitions. It is to this, \emph{co-event interpretation} that we now turn.

\subsection{The Co-Event Interpretation}\label{sec:the coevent interpretation}

\subsubsection{Co-Events for Classical Theories}\label{sec:classical coevents}

When $(\O,\EA,\P)$ is a classical theory, the `standard interpretation' identifies exactly one of the elements of $\O$ as the `real history', $\g_r$, which `actually occurs'. However the dynamics (in the form of the measure $\P$) does not uniquely identify $\g_r$; the most we can say with complete confidence is that the dynamics should not \emph{preclude} (rule out) the real history, which would occur were $\g_r$ to be an element of a null set (a set of measure zero). We thus use the term \emph{potential reality} in referring to a history which is not an element of a null set, and denote the set of such histories the \emph{space of potential realities}. To avoid confusion with the potential realities, we will henceforth refer to the real history $\g_r$ as the \emph{actually real history}.

Recalling that we can think of the event algebra $\EA$ as the set of propositions that can be made about the system, notice that knowledge of which potential reality is actually real allows us to give true/false answers to each of these propositions. More concretely, we say that an event is `true' if it contains the actually real history $\g_r$, otherwise it is `false'. Identifying our space of truth valuations with $\Z2$ (where $1$ is associated with `yes'/`true' and $0$ with `no'/`false'), we can use the history $\g_r$ to construct a truth valuation map, or \emph{co-event}:
\bea
\g_r^*:\EA &\rightarrow& \Z2, \nonumber \\
\g_r^*(A) &=& \left\{\ba{cc} 1 & \g_r\in A \\ 0 & \g_r \not\in A. \ea\right. \label{eq:classical coevent action}
\eea
We can then phrase our requirement that no potential reality should lie in a null set in the language of truth valuation maps, so that $\g\in\O$ is a potential reality whenever
\beq\label{eq:classical preclusion}
\P(A) = 0 \Rightarrow \g^*(A) = 0,
\eeq
a condition which we will refer to as \emph{preclusion}. We will refer to co-events obeying equation \ref{eq:classical preclusion} as \emph{preclusive}.

Sorkin has suggested a subtle shift in thinking whereby the co-events $\g^*$ should be regarded as real rather than the histories $\g$ \cite{Sorkin:2006wq,Sorkin:2007,Dowker:2007zz}. Thus given a classical histories theory $(\O,\EA,P)$ our space of potential realities is now the set of preclusive $\g^*$ where $\g\in\O$.

This shift in thinking is non-trivial because co-events have a natural algebraic structure which we can use to generalise them from the classical case. As previously noted, $\EA$ is both a ring and an algebra over $\Z2$ with addition identified with symmetric difference and multiplication with intersection. As a field, $\Z2$ is also a ring (and could be thought of as an algebra over $\Z2$), inheriting addition and multiplication from $\ZZ$. When the sample space is finite, it can be shown that the co-events $\g^*$ are precisely the (ring) homomorphisms from $\EA$ to $\Z2$.

\begin{lemma}\label{lemma:classical coevents are homomorphisms}
Let $\H$ be a histories theory with a finite sample space. Then
$$\{\g^*|\g\in\O\} = Hom(\EA,\Z2),$$
where $Hom(\EA,\Z2)$ is the set of (ring) homomorphisms from $\EA$ to $\Z2$ excluding the zero map.
\end{lemma}

Thus the co-events $\g^*$ can be characterized algebraically by the fact that they are (non-zero) homomorphisms, which is equivalent to their adherence to the following `rules':
\bea
\g^*(A+B) &=& \g^*(A) + \g^*(B) ~~(linearity), \label{eq:linearity} \\
\g^*(AB) &=& \g^*(A)\g^*(B) ~~(multiplicativity), \label{eq:multiplicativity}
\eea
for all $A,B\in\EA$.

\subsubsection{Co-Events for Quantum Theories}\label{sec:coevents for quantum theories}

The co-events introduced in the previous section are tied to individual histories. We can generalise this structure to encompass a wider class of answering map.

\begin{definition}\label{def:coevent}
Let $\H$ be a histories theory. Then a map
$$\p:\EA\rightarrow\Z2,$$
is a \textbf{co-event} if it maps the empty set to zero. We denote the space of co-events by $\CE$. A co-event is \textbf{preclusive} if
$$\mu(A)=0\Rightarrow\p(A)=0~~\forall~A\in\EA.$$
\end{definition}

It is clear that when the underlying histories theory is classical, the above definition of preclusiveness reduces to the one given in equation \ref{eq:classical preclusion}. Because of this, we can define the `classical' co-events within this more general setting, describing any co-event (other than the zero map) as \emph{classical} if it obeys linearity \& multiplicativity (equations \ref{eq:linearity} \& \ref{eq:multiplicativity}). Lemma \ref{lemma:classical coevents are homomorphisms} means that the classical co-events are all the co-events of the form $\g^*$ for some $\g\in\O$, as described above. Given a histories theory $\H$, we refer to the set of preclusive classical co-events as $\C$.

Perhaps the most obvious means of applying co-events to quantum measure theories would be to use the classical structure unaltered; thus given a histories theory $\H$ our space of potential realities would be $\C$. However it is possible to show that classical co-events are an inadequate ontology for quantum dynamics; in particular the Kochen-Specker theorem \cite{Kochen:1967}, often cited as an obstruction to realism, and its proof by Asher Peres \cite{Peres:1991}, can be used to construct a `gedanken-experimentally realisable' system that can not be described by any classical co-event \cite{Dowker:2007zz}. This leads to:

\begin{theorem}\label{thm:PKS}
There exists a `gedanken-experimentally realisable' histories theory $\H$ with a finite sample space such that
$$\C=\emptyset.$$
\end{theorem}
\begin{proof}
See \cite{Dowker:2007zz}.
\end{proof}

The failure of classical co-events to adequately describe quantum measure theories has led to a search for alternatives. A number of co-event `\emph{schemes}' have been proposed, typically generalising from classical co-events by omitting and/or relaxing one or both of equations \ref{eq:linearity} \& \ref{eq:multiplicativity}. Many of these proposals, such as the \emph{linear} \cite{Sorkin:2006wq}, \emph{quadratic} and \emph{general polynomial} schemes, have been conclusively ruled out \cite{CoEventSchemes}. Other proposals, such as the \emph{ideal} and \emph{multiplicative} schemes remain as viable candidates \cite{Sorkin:2007,CoEventSchemes}. For various reasons \cite{Sorkin:private,Dowker:private,CoEventSchemes}, the multiplicative approach has been adopted as the current working model of co-event theory, and it is to this scheme that we now turn.

\subsubsection{The Multiplicative Scheme}\label{sec:the multiplicative scheme}

Starting with a histories theory $\H$, a co-event $\p\in\CE$ is \emph{multiplicative} if it is not the zero map and it obeys the multiplicativity condition (equation \ref{eq:multiplicativity}):
$$\p(AB) = \p(A)\p(B)~~\forall~A,B\in\EA.$$
An important feature of multiplicative co-events is their description in terms of filter theory:
\begin{lemma}
Let $\H$ be a histories theory with a finite sample space, and let $\p\in\CE$ be a multiplicative co-event. Then $\p^{-1}(1)$ is a principal filter.
\end{lemma}
\begin{proof}
Let $A \in \p^{-1}(1)$. If $A \subset B$ then $AB = A$ and therefore $1 = \p(A) = \p(AB) = \p(A)\p(B) = \p(B)$, which means that $B \in \p^{-1}(1)$. Further, if $A,B\in\p^{-1}(1)$ then $\p(AB)=\p(A)\p(B)=1$, so $A\cap B\in\p^{-1}(1)$; thus $\p^{-1}(1)$ is a filter.

Now because $\O$ and therefore $\EA$ are finite $\p^{-1}(1)$ must contain a minimal element (under the partial order defined by set inclusion), which we will denote $\p^*\in\EA$. Further $\p^*$ must be unique, otherwise if $A\in\p^{-1}(1)$ is also minimal then $A\cap\p^*$ is in the filter $\p^{-1}(1)$ and is contained in both $A$ and $\p^*$, contradicting the assumption that they are minimal. Therefore $\p^{-1}(1)$ is a principal filter and $\p^*$ its principal element.
\end{proof}

The association of a multiplicative co-event $\p$ with the principal element $\p^*$ of the related filter $\p^{-1}(1)$ can be described as a map:
\bea
*:\CE &\rightarrow&\EA, \nonumber \\
*:\p &\mapsto& \p^*.
\eea
We can extend $*$ to an involution on $\EA\times\CE$ by defining
\bea
*:\EA &\rightarrow&\CE, \nonumber \\
*:A &\mapsto& A^*,
\eea
where for all $B\in\EA$
\beq
A^*(B) = \left\{\begin{array}{cc} 1 & {\text{if} ~A \subset B} \\
                                  0 & {\text{otherwise}}\end{array}\right.
\eeq
It is easy to see that $*$ is bijective, and that $(\p^*)^*=\p$, $(A^*)^*=A$. We can therefore think of $*$ as a duality, and will call $\p^*$ the \emph{dual} of $\p$ (and $A^*$ the dual of $A$).

Note that every classical co-event is also a multiplicative co-event, and that $\{\g\}^*=\g^*$, so our definition of the $*$ map is consistent with our previous notation. Now the classical co-events $\g^*$ are dual to single histories $\g$, which are the potential realities (given that they are preclusive) in classical theories. This suggests that we might think of a multiplicative co-event $A^*$ as an \emph{ontological coarse graining}; rather than one history $\g$ (or co-event $\g^*$) describing reality we now have a multiplicity of histories $A$ (or the co-event $A^*$), and only the properties common to all of these histories are true.

To generalise from (preclusive) classical to (preclusive) multiplicative co-events we have dropped the linear rule (equation \ref{eq:linearity}), thus increasing the number of `allowed' co-events. Unfortunately we now have `too many' co-events; in particular given a classical measure we could have non-classical multiplicative co-events, implying non-classical behaviour of dynamically classical systems. Indeed, if the measure does not admit null sets then the dual of every event will be preclusive and multiplicative. We therefore need a further constraint that will restrict us to classical co-events given a classical measure. We say that a multiplicative co-event $\psi$ \emph{dominates} a multiplicative co-event $\p$ if
$$\p(A)=1\Rightarrow\psi(A)=1~~~\forall~A\in\EA.$$
A preclusive multiplicative co-event is \emph{primitive} (among the set of preclusive multiplicative co-events) if it is not dominated by any other preclusive multiplicative co-event. We can also motivate primitivity using the natural order on filters. A preclusive multiplicative co-event $\p$ is \emph{primitive} if and only if the associated filter $\p^{-1}(1)$ is maximal among the filters corresponding to preclusive multiplicative co-events. Primitivity can be thought of as a condition of ``maximal detail'' or ``finest graining'' consistent with preclusion. We are now in a position to fully specify the concept of potential reality in the multiplicative scheme,
\begin{definition}
Let $\H$ be a histories theory with a finite sample space. Then we denote the set of primitive preclusive multiplicative co-events by $\M$.
\end{definition}
Comparing the multiplicative scheme with our previous results for classical theories, it is $\M$ which replaces $\C$ as the space of potential realities. In this scheme individual histories (or equivalently classical co-events) may no longer be potentially real, and thus can not be expected to necessarily obey the logic describing reality. Notice that we use `standard' (Boolean) logic in our reasoning concerning the co-events themselves.

Recall that classical co-events were discarded because we could find gedanken-experimentally realisable histories theories that no preclusive classical co-event could describe (theorem \ref{thm:PKS}). We are therefore reassured by the following,

\begin{lemma}\label{lemma:mult existence primitivity}
Let $\H$ be a history theory with a finite sample space and let $A\in\EA$ be non-negligible (ie $A$ is not a subset of a null set). Then $\exists \p\in\M$ such that $\p(A)=1$.
\end{lemma}
\begin{proof}
Assume that $\nexists~\p\in\M$ such that $\p(A)=1$. Because $A$ is non-negligible we know that $A^*$ is preclusive. Then since $A^*(A)=1$, by assumption $A^*$ can not be primitive, so there exists some preclusive $\psi_1$ dominating $A^*$, and by the definition of domination $\psi_1(A)=1$. The same argument shows that $\psi_1$ can not be primitive, so there exists a preclusive $\psi_2$ dominating $\psi_1$ with $\psi_2(A)=1$. Carrying on in this fashion we can find an infinite sequence of preclusive co-events $\{\psi_i\}_{i=0}^\infty$ with $\psi_0=A^*$, with $\psi_i$ dominating $\psi_{i-1}$ (and thus $\psi_{i-n}$) and with $\psi_i(A)=1$. However this contradicts our assumption that $\O$, and therefore $\EA$ and $\CE$, are finite; hence the result.
\end{proof}

Note that $\mu(\O)=1$ in any histories theory, thus whenever $\O$ is finite we know that $\M\neq\emptyset$. Further, we can show that primitivity is indeed successful in ensuring that given a classical measure our space of potential realities consists of classical co-events.

\begin{lemma}\label{lemma:weak emergent classicality}
Let $\H$ be a history theory with a finite sample space and a classical (level $1$) measure, then $\M=\C$.
\end{lemma}
\begin{proof}
Let $\p\in \M$. The Kolmogorov sum rule implies that the union, $Z$, of null sets in $\EA$ is itself null, so the preclusivity of $\p$ implies $\p^*\not\subset Z$. Thus $\exists \g \in \p^*\setminus Z$ such that  $\g$ is not an element of any null set. This means $\g^*$ is preclusive, but then $\g\in\p^*$ means that either $\g^*$ dominates $\p$, or that $\p=\g^*$. But by assumption $\p$ is primitive, and so is not dominated by any preclusive co-event. Therefore $\p=\g^*$, a homomorphism, and $\p\in\C$. Further, any homomorphism is a multiplicative co-event, thus $\M=\C$.
\end{proof}

Now, as touched on in section \ref{sec:dynamical classicality}, the histories approach typically assumes that dynamical classicality is an `emergent' phenomena, so that classical histories theories are to be thought of as coarse grainings of `deeper' quantum measure theories. Thus the `measure of the universe' is level $2$ whereas we access it through dynamically classical partitions of the sample space on which the induced measure is level $1$. We would therefore like to see a stronger version on lemma \ref{lemma:weak emergent classicality} which ensured that (given $\H$) all $\p\in\M$ `behave classically' on `classical partitions'. Of course whether or not we can achieve such a result depends upon how we define these two terms, `classical behaviour' turns out to be the simpler of the two:

\begin{definition}
Let $\H$ be a histories theory with a finite sample space and let $\la$ be a partition of $\O$. Then we say that a co-event $\p\in\CE$ is \textbf{classical on} $\la$ if the restriction of $\p$ to the subalgebra of $\EA$ generated by $\la$ is a homomorphism.
\end{definition}

As discussed in section \ref{sec:dynamical classicality}, defining exactly what we mean by a `classical partition' remains an open area of research. On the one hand given a histories theory $\H$ we would like every $\p\in\M$ to be classical on every classical partition of $\O$, on the other hand we want to be able to regard classical partitions as `dynamically classical' in some sense; at the very least partitions corresponding to experimentally observable alternatives should be classed as classical. Perhaps the most obvious way forward would be to regard all decoherent partitions as classical, however this condition turns out to be too strong. At the very least, because we observe and reason with Boolean logic we require our potentially real co-events $\p\in\M$ to treat classical partitions classically. We call this condition \emph{ontological classicality} or \emph{classicality with respect to $\M$} to distinguish it from both dynamical classicality and the more general notion of a classical partition, which we may require to be stronger.

\begin{definition}
Let $\H$ be a histories theory with a finite sample space. We say that a partition $\la$ of $\O$ is \textbf{classical with respect to $\M$} if every $\p\in\M$ is classical on $\la$.
\end{definition}

In appendix \ref{appendix:principle classical partition} we demonstrate the existence of a \emph{principle classical partition} $\la_C$ which all $\p\in\M$ treat classically; further any partition on which any $\p\in\M$ is classical is a sub-partition of $\la_C$.

Relating this back to the dynamics, Sorkin has proposed a new dynamical condition, \emph{preclusive separability} \cite{Sorkin:private}, that ensures classicality with respect to $\M$. A partition $\la$ of $\O$ is preclusively separable if for any null set $Z\in\EA$ and any partition element $A\in\la$, either $Z\cap A=\emptyset$ or $\exists$ a null set $Z_A\subset A$ such that $Z_A\supset A\cap Z$.

The class of preclusively separable partitions includes all partitions in which any two histories belonging to distinct partition elements `end' at different locations at the final time (or equivalently correspond to different final time projectors in their class operators, if the histories theory is derived from a Hilbert space theory). Thus in any histories theories derived from a Hilbert space theory, any partition corresponding to the alternative outcomes of a particular Copenhagen measurement (thus corresponding to `orthogonal projectors') obeys preclusive separability \cite{Sorkin:private,Dowker:private}, so the fact that preclusively separable partitions are always classical with respect to $\M$ is a powerful result. This remains an open area of research.

\subsubsection{Co-Events for Coins}

We illustrate the use of co-events with a simple example. Consider a single toss of a coin; we assume that the outcome will be either `heads' (`h') or `tails' (`t'), so our sample space is:
$$\O = \{h,t\},$$
and our event algebra
$$\EA = P\O = \{\emptyset,\{h\},\{t\},\O\}.$$
We assume no interference between the two outcomes, so that our measure is classical
\bea
\mu(\emptyset) &=& 0, \nonumber \\
\mu(\{h\}) &=& p, \nonumber \\
\mu(\{t\}) &=& 1-p, \nonumber \\
\mu(\O) &=& 1. \nonumber
\eea
Then assuming that $p\not\in\{0,1\}$ we have
$$\C = \{h^*,t^*\},$$
and the multiplicative co-events are:
$$h^*,t^*,\O^*.$$
Notice that $\emptyset^*$ is not included because it is not a co-event (since no co-event maps $\emptyset$ to unity). Now assuming that $p\not\in\{0,1\}$ all three of the above co-events are preclusive, however both $h^*$ and $t^*$ dominate $\O^*$, while being themselves primitive. Then we have
$$\M = \C,$$
as we would expect from lemma \ref{lemma:weak emergent classicality}. Nevertheless, under non-classical measures non-classical co-events similar to $h^*t^*$ are often primitive. It is thus instructive to delve further into the non-classicality of $\O^*$, which we can do by constructing its `truth table':
$$
\begin{array}{c|cccc}
 & \emptyset & \{h\} & \{t\} & \O \\
\hline
\O^* & 0 & 0 & 0 & 1
\end{array}
$$
where each table entry is the truth valuation $\O^*(A)$ of the relevant event $A$. We can rephrase this in `question \& answer' format:
\\
\\
\begin{tabular}{c|c}
\bf{Question} & \bf{Answer} \\
\hline
Does the coin show heads? & no \\
Does the coin show tails? & no \\
Does the coin show one of heads or tails? & yes \\
\end{tabular}
\\
\\
Noting that $\{h,t\}$ can be thought of as a (trivial) partition of $\O$, we could say that $\O^*$ is not classical on the `classical partition' $\{h,t\}$. This is because $\O$ has non-empty intersection with both partition elements.

\subsection{The Goals of this Paper}\label{sec:goals}

As far as co-events go our predictive power is (so far) limited. Given an event algebra and a measure, our co-event approach identifies the primitive preclusive co-events (assuming a particular scheme) as the potential realities of our theory; however we cannot in general reconstruct the measure from the set of potentially real co-events. Were we to ignore the measure and use only the potentially real co-events themselves (without knowledge of which co-event is actually real), we would be unable to make any predictive statements beyond ruling out as impossible those events that are precluded (mapped to zero) by all the potentially real co-events, and identifying as certain all those events that are mapped to unity. This may not be quite as restrictive as it appears, for by \emph{conditioning} on known truth valuations of events (disregarding the co-events that disagree with these valuations) we may be able to narrow down the set of potentially real co-events to the point that such predictions are relatively powerful \cite{Sorkin:private}.

However there is reason to believe (depending on the definition of classicality we adopt) that conditioning on previous observations will not increase our power to predict other observable events; though conditioning on non-observable events may increase our ability to predict observable events and conditioning on observable events may increase our ability to predict non-observable events.

In any case, it is clear that such conditioning will not in general allow us to recover the \emph{entirety} of either the measure or its predictive content, whatever definition of classicality we adopt; whereas we would like to be able to express the full dynamical \& predictive content of a quantum measure theory in terms of the potentially real co-events rather than the histories. We would thus be able to `start with' a set of potentially real co-events and make our predictions using only these co-events and structures defined in terms of them. In this paper we will explore attempts to achieve this; we are motivated in our endeavor by three arguments.

Firstly, we are treating the primitive co-events themselves as the potential realities of the histories approach. The authors feel that we should be able to express all meaningful statements concerning a theory in terms of the theory's potential realities; even if \emph{in practise} it is simpler to use the measure $\mu$ defined on the event algebra, at least \emph{in principle} this structure should `emerge' from a `deeper' expression of the dynamical \& predictive content of the theory in terms of the co-events themselves.

Secondly, the presence of `exactly null' sets in many gedankenexperimental frameworks (such as the double slit experiment) depends upon idealised assumptions (such as point slits) that do not hold in practice. Regardless of our interpretation, co-events as they now stand cannot distinguish between events of large or small measure, so our use of null sets to approximate `almost null' sets in idealised gedankenexperiments requires justification. As mentioned above, at present the range of predictive statements we can make with the set of potentially real co-events is narrow; were we able to fully express the predictive content of quantum measure theory in terms of the potentially real co-events we may be able to better grapple with issues such as `almost null sets' that are phrased in terms of probability rather than preclusion.

Thirdly, were we able to express in terms of co-events the predictive content of those histories theories derived from Hilbert space theories, we might then be in a position to extrapolate our findings and better understand the predictive content of a general quantum measure theory (in which, for example, there may be no `external' measuring apparatus).

In what follows we will explore two recent attempts to rephrase in terms of co-events the dynamical \& predictive content of quantum measure theories. We first consider attempts to move wholesale from the `histories theory' structure to one of `co-event theories' in which the dynamics is defined directly in terms of co-events (as anticipated in \cite{Sorkin:2007}). We then turn our attention to the important concept of `approximate preclusion', and discuss the details of its implementation.

Unless specifically noted otherwise we will assume that all sample spaces are finite, and that all histories theories are derived from Hilbert space theories (we will continue to make the assumptions described above regarding observable events). Further, as previously stated, given a histories theory we will assume that we can identify the observable events, that each observable event participates in an observable partition, and that all observable partitions are dynamically classical.

\section{Dynamics on co-events}\label{sec:dynamics on coevents}

\subsection{Placing a Measure on the Space of Co-Events}

One approach, anticipated by Sorkin in \cite{Sorkin:2007}, is to rephrase the dynamics (the quantum measure $\mu$) in terms of the co-events themselves; for example by placing a measure on the space of co-events. This would be a natural development, given that the co-events are the basic potential realities of the theory, and could lead to the rephrasing of the entirety of quantum mechanics in terms of co-events, bypassing our current need to construct dynamics upon the event algebra. We could then explain the rarity of low probability events (such as `almost null' sets) by using the dynamics on our co-events to suppress (assign a low probability to) those co-events that value them to $1$.

Perhaps the most obvious way to do this would involve literally `moving' the dynamics from the event algebra onto the co-events, which can thus realise their  status as the central objects of the theory. So far we have always begun with a histories theory $\H$ from whence we have derived our allowed co-events ${\cal{S}}\H$. Since we are asserting that ${\cal{S}}\H$ rather than $\O$ is the `true' sample space of `potential realities' it would seem more natural to place the co-events at the center of our structure, moving from a histories theory $\H$ to a full \textit{co-event theory} $(S,PS,\P_S)$, where $S$ is a set of (allowed) co-events, for example $S=\M$, $PS$ is the power set of $S$, and $\P_S$ is a probability measure on $PS$. We have assumed a probability measure, for if we were to use a more general higher order measure we would not have gained anything from our adoption of co-events in place of histories, encountering the same interpretational difficulties that led us to co-events. Indeed, following that route might even lead us to adopt co-co-events, co-co-co-events and so on. Conversely, since we are assuming that `one co-event is real', the use of a probability measure echoes the use of probability in classical stochastic theories in which we believe that `one history is real'.

In a sense this approach is born of the Bayesian approach to probability \cite{Ramsey:1926,Finetti:1930,Cox:2001}, which looks on probability as representing a degree of information concerning a system and thus ascribes meaning to individual probability statements which can therefore be treated as `logical' predicates in their own right \cite{Cox:2001}. Following this philosophy we seek to express our incomplete information about the actual reality in terms of probability statements concerning the potential realities, much as we are accustomed to do in the classical theory. We can, therefore, think of $\P_S(\{\p\})$ as the probability that the co-event $\p$ is the actually real co-event; consequently $\P_S(\{\p\})=0$ could be interpreted to mean that $\p$ is not a potential reality. Indeed, we could even enforce our choice of $S$ as the set of potential realities by extending the domain of $\P_S$ to all of $\CE$, and constraining $\P_S$ to be zero on all co-events outside $S$.

Nevertheless our observations are in terms of events in $\EA$, so we must relate this new structure back to our Copenhagen probabilistic predictions. As stated above we are focusing only on histories theories $\H$ that are derived from a Hilbert space theory; we have assumed that in such theories we can identify the set of `observable events', that each observable event is an element of an `observable partition', and that $\mu$ is classical on all such partitions. As noted above, if $A$ is an observable event, $\mu(A)$ can be interpreted as a probability, and this probability agrees with the Copenhagen predicted probability for $A$. Then at the very least we want, for all observable $A$,

\beq\label{eq:binary probs on coevents}
\mu(A)=\left\{\begin{array}{c}1\\0\end{array}\right. \Rightarrow\P_S(\{\p\in S|\p(A)=1\})=\left\{\begin{array}{c}1\\0.\end{array}\right.
\eeq

As we will discuss below, non-binary (ie not equal to $0$ or $1$) probabilities themselves require interpretation, and can arguably be defined in terms of repeated trials and binary probabilities. However we will steer clear of this issue for now, and instead of using  equation \ref{eq:binary probs on coevents} to define probabilities through repeated trials we will instead use the stronger:
\beq\label{eq:probs on coevents weak}
\P_S(\{\p\in S|\p(A)=1 \})=\mu(A),
\eeq
for all observable events $A$. However it is not clear how `observable events' should be defined in our quantum measure theory (see section \ref{sec:dynamical classicality}). We sidestep this question by simply strengthening the condition once more to require:
\beq\label{eq:probs on coevents}
\P_S(\{\p\in S|\p(A)=1 \})=\mu(A) ~\forall A\in\EA.
\eeq
One difficulty with this condition is that in general $\mu(A)$ can exceed $1$. However, even when we are able to constrain $\mu$ to take values in $[0,1]$, equation \ref{eq:probs on coevents} is a problematic condition for multiplicative co-events, as the following theorem shows.
\begin{theorem}\label{thm:dynamical coevents are quadratic}
Let $\H$ be a histories theory with a finite sample space, let $S\subset\CE$ be a set of multiplicative co-events on $\EA$, so that
$$\p(AB)=\p(A)\p(B)~\forall A,B\in\EA,~\forall\p\in S.$$
Further, let $\P_S$ be a probability measure on $S$ such that
$$\P_S(\{\p\in S|\p(A)=1 \})=\mu(A) ~\forall A\in\EA.$$
Then for any $\p\in S$, $\P_S(\{\p\})=0$ unless for all $A,B,C\in\EA$
\bea
\p(A+ B+ C) &=& \p(A+ B)+\p(B+ C)+\p(C+ A) \nonumber \\
&& +\p(A)+\p(B)+\p(C). \label{eq:quadratic rule}
\eea

\end{theorem}

As noted above, the use of a probability measure on the space of co-events echoes its use on histories in classical stochastic theories. Then following our usual interpretation that dynamically precluded `events' do not occur, we conclude that the actually real co-event must obey equation \ref{eq:quadratic rule}. Now equation \ref{eq:quadratic rule} is the defining property of the \emph{quadratic (co-event) scheme} \cite{CoEventSchemes} we alluded to in section \ref{sec:coevents for quantum theories}, a \emph{quadratic co-event} being in general a co-event obeying equation \ref{eq:quadratic rule} though not necessarily obeying linearity or indeed multiplicativity. Then since the quadratic scheme has been ruled out \cite{CoEventSchemes}, this puts into question our idea of `moving' the dynamics onto the co-events by constructing `co-event theories' in place of histories theories, unless we are willing to abandon the multiplicative scheme.

Of course this approach is not entirely ruled out, to arrive at equation \ref{eq:probs on coevents} we have made several simplifying assumptions so there may still be some scope for further investigation. For example, if we accept $\M$ as the space of potential realities of a histories theory $\H$, we might attempt to define all other elements of the theory in terms of the co-events $\p\in\M$. In particular, we could define events as maps from $\M$ to $\Z2$ using the `dual map' $A[\p]=\p(A)$. Since we may not be able to distinguish between all elements of $\EA$ in this manner, we might perhaps choose to take equivalence classes of `distinguishable events' ($A\sim B$ if $A[\p]=B[\p]$ for all $\p\in\M$), and to attempt to define a measure on this set as a stepping stone to a `full co-event theory'. Such a structure is yet to be explored. Alternatively, we might explore the restrictions on $\mu$ that would enable equation \ref{eq:probs on coevents}, or we might turn to another scheme.

\subsection{Proof of Theorem \ref{thm:dynamical coevents are quadratic}}

In this section we will prove theorem \ref{thm:dynamical coevents are quadratic}. We will need the following technical lemmas.

\begin{lemma}\label{lemma:quantum sum rule implies quadratic rule}
Let $\H$ be a histories theory with a finite sample space and let $\p$ be a (not necessarily multiplicative) co-event. Then
\bea
\p(A\sqcup B\sqcup C) &=& \p(A\sqcup B)+\p(B\sqcup C)+\p(C\sqcup A) \nonumber \\
&& +\p(A)+\p(B)+\p(C),\label{eq:proof sum rule equivalent to quadratic rule 1}
\eea $\forall~disjoint~A,B,C\in\EA$ if and only if \bea
\p(A'+ B'+ C') &=& \p(A'+ B')+\p(B'+ C')+\p(C'+ A') \nonumber \\
&& +\p(A')+\p(B')+\p(C'),\label{eq:proof sum rule equivalent to quadratic rule 2} \eea $\forall A',B',C'\in\EA$.
\end{lemma}
\begin{proof}
First note that equation \ref{eq:proof sum rule equivalent to quadratic rule 1} is simply the restriction of equation \ref{eq:proof sum rule equivalent to quadratic rule 2} to disjoint sets, so that we immediately see that equation \ref{eq:proof sum rule equivalent to quadratic rule 1} $\Leftarrow$ equation \ref{eq:proof sum rule equivalent to quadratic rule 2}.

To prove the converse we consider sets $A_1,A_2,A_3\in\EA$ and break down their union $\bigcup_{i=1}^3A_i$ into disjoint components by defining: \bea
A_1A_3A_3 &=& A_1\cap A_2\cap A_3 \nonumber \\
\overline{A_iA_j}&=&(A_i\cap A_j)\setminus (A_1A_3A_3) \nonumber \\
\overline{A_i} &=& A_i\setminus (\overline{A_iA_j}\sqcup\overline{A_iA_k}\sqcup A_1A_3A_3). \nonumber \\
\eea Then \beq\nonumber \bigcup_iA_i = (\bigsqcup_i\overline{A_i})\sqcup(\bigsqcup_{i<j}\overline{A_iA_j})\sqcup A_1A_3A_3. \eeq
Assuming equation \ref{eq:proof sum rule equivalent to quadratic rule 1}, we can use this notation to decompose the terms in equation \ref{eq:proof sum rule equivalent to quadratic rule 2}. The left hand side becomes
\bea
\p(A_1+ A_2 + A_3) &=& \p((\bigsqcup_i\overline{A_i})\sqcup A_1A_2A_3) \nonumber \\
&=& \sum_{i<j}\p(\overline{A_i}\sqcup\overline{A_j}) + \sum_i \p(\overline{A_i}\sqcup A_1A_2A_3). \nonumber
\eea
Similarly turning our attention to the right hand side, for $i,j,k$ distinct we find that
\bea
\p(A_i+ A_j) &=& \p(\overline{A_i}\sqcup\overline{A_iA_k}\sqcup\overline{A_j}\sqcup\overline{A_jA_k}) \nonumber \\
&=& \p(\overline{A_i}\sqcup \overline{A_j})+\p(\overline{A_i}\sqcup \overline{A_iA_k})+ \p(\overline{A_i}\sqcup \overline{A_jA_k}) \nonumber \\
&& + \p(\overline{A_j}\sqcup \overline{A_iA_k})+ \p(\overline{A_j}\sqcup \overline{A_jA_k}) + \p(\overline{A_iA_k}\sqcup \overline{A_jA_k}), \nonumber
\eea
and
\bea
\p(A_i) &=& \p(\overline{A_i}\sqcup\overline{A_iA_j}\sqcup\overline{A_iA_k}\sqcup A_1A_3A_3) \nonumber \\
&=& \p(\overline{A_i}\sqcup\overline{A_iA_j})+\p(\overline{A_i}\sqcup\overline{A_iA_k})+\p(\overline{A_i}\sqcup A_1A_3A_3) \nonumber \\
&& + \p(\overline{A_iA_j}\sqcup\overline{A_iA_k})+ \p(\overline{A_iA_j}\sqcup A_1A_3A_3)+ \p(\overline{A_iA_k}\sqcup A_1A_3A_3). \nonumber
\eea
Comparing these it is easy to see that that the left and right hand sides of equation \ref{eq:proof sum rule equivalent to quadratic rule 2} are equal, hence the result.
\end{proof}

We now introduce some notation. First, given a histories theory $\H$ and a co-event $\p\in\CE$ we define
\bea
Q_{ABC}(\p) &=& \p(A\sqcup B\sqcup C) + \p(A\sqcup B)+\p(B\sqcup C)+\p(C\sqcup A) \nonumber \\
&& +\p(A)+\p(B)+\p(C),
\eea
where $A,B,C\in\EA$ are disjoint events. Then lemma \ref{lemma:quantum sum rule implies quadratic rule} tells us that $\p\in\CE$ is quadratic if and only if $Q_{ABC}(\p)=0$ for all disjoint $A,B,C\in\EA$.

Now though co-events take values in $\Z2$, probability measures take values in $\R$. It will be useful to algebraically combine the images of (the values returned by) co-events and probability measures, which we can achieve by thinking of $\Z2$ as a subset of $\R$; thus given a co-event $\p\in\CE$ we can construct the related map:
\bea
\tilde{\p}:\EA&\rightarrow&\R \nonumber \\
\tilde{\p}(A) &=& \left\{\ba{cc} 0 & \p(A)=0 \\ 1 & \p(A) = 1. \ea\right. \nonumber
\eea
Using this map we can define a real valued analogue of the function $Q_{ABC}$:
\bea
R_{ABC}:\CE &\rightarrow& \R \nonumber \\
R_{ABC}(\p) &=& \tilde{\p}(A\sqcup B\sqcup C)- \tilde{\p}(A\sqcup B)-\tilde{\p}(B\sqcup C)-\tilde{\p}(C\sqcup A)  \nonumber \\
&& +\tilde{\p}(A)+\tilde{\p}(B)+\tilde{\p}(C)) \nonumber,
\eea
where $A,B,C\in\EA$ are disjoint events. We can show that if $\p$ obeys the multiplicative rule then $R_{ABC}$ is either $0$ or $1$, regardless of our choice of $A,B,C$.
\begin{lemma}\label{lemma:dynamical multiplicative coevents R is 0 or 1}
Let $\H$ be a histories theory with a finite sample space, and let $\p\in\CE$ be multiplicative. Then for all disjoint $A_1,A_2,A_3\in\EA$ we have:
$$R_{A_1A_2A_3}(\p)\in\{0,1\}.$$
\end{lemma}
\begin{proof}
Since $\p$ is multiplicative, for any $A\in\EA$ we know that $\p(A)=1$ if and only if $\p^*\subset A$. Then fixing our three events $A_1,A_2,A_3$ we have four cases:
\begin{enumerate}
  \item $\p^* \not\subset A_1\sqcup A_2\sqcup A_3$; then $\p( A_1\sqcup A_2\sqcup A_3)=\p( A_i\sqcup A_j)=\p(A_k)=0$ for all $i,j,k$, so that $R_{A_1A_2A_3}(\p)=0$.

  \item $\p^* \subset A_1\sqcup A_2\sqcup A_3$, $\p^*\not\subset A_i\sqcup A_j$ for any $i,j$; then $\p( A_1\sqcup A_2\sqcup A_3)=1$ but $\p( A_i\sqcup A_j)=\p(A_k)=0$ for all $i,j,k$, so that $R_{A_1A_2A_3}(\p)=1$.

  \item $\p^* \subset A_i\sqcup A_j$ for some $i,j$, but $\p^*\not\subset A_k$ for any $k$. Without loss of generality we can assume that $i,j=1,2$; then $\p( A_1\sqcup A_2\sqcup A_3)=1$, $\p(A_1\sqcup A_2) = 1$, and $\p(A_k)=0$ for all $k$. Further, $\p^*\not\in A_2\sqcup A_3$ otherwise $\p^*\in A_2$, similarly $\p^*\not\in A_3\sqcup A_1$, so $\p(A_2\sqcup A_3) = \p(A_3\sqcup A_1) = 0$. Therefore $R_{A_1A_2A_3}(\p)=1-1=0$.

  \item $\p^*\subset A_i$ for some $i$; without loss of generality assume that $i=1$. Then $\p^*\not\subset A_2,A_3$, since these sets are disjoint to $A_1$. Then $\p(A_k)=\delta_{1k}$, $\p(A_1\sqcup A_2)=\p(A_1\sqcup A_3)=1$, $\p(A_2\sqcup A_3)=0$, and finally $\p(A_1\sqcup A_2 \sqcup A_3)=1$. Thus $R_{A_1A_2A_3}(\p)=1-1-1+1=0$.
\end{enumerate}
Hence the result.
\end{proof}

Finally, we can relate $R_{ABC}$ and $Q_{ABC}$.

\begin{lemma}\label{lemma:dynamical coevents R and Q}
Let $\H$ be a histories theory with a finite sample space and let $\p\in\CE$ be a multiplicative co-event. Then if $A,B,C\in\EA$ are disjoint
$$R_{ABC}(\p)=0 \Rightarrow Q_{ABC}(\p)=0.$$
\end{lemma}
\begin{proof}
First, denote by $X~Mod~2$ the $\Z2$ element corresponding to the integer $X~mod~2$, where $X$ is an integer in $\R$. Then $\p(A)=\tilde{\p}(A)~Mod~2$ for any $A\in\EA$, so noting that $+$ and $-$ are equivalent in $\Z2$, for any disjoint $A,B,C\in\EA$ we have:
\bea
R_{ABC}(\p)~Mod~2 &=& \tilde{\p}(A\sqcup B\sqcup C)~Mod~2 \nonumber \\
&& - \tilde{\p}(A\sqcup B)~Mod~2 -\tilde{\p}(B\sqcup C)~Mod~2-\tilde{\p}(C\sqcup A)~Mod~2  \nonumber \\
&& +\tilde{\p}(A)~Mod~2+\tilde{\p}(B)~Mod~2+\tilde{\p}(C)~Mod~2, \nonumber \\
&=& \p(A\sqcup B\sqcup C) + \p(A\sqcup B)+\p(B\sqcup C)+\p(C\sqcup A) \nonumber \\
&& +\p(A)+\p(B)+\p(C), \nonumber \\
&=& Q_{ABC}(\p). \nonumber
\eea
Hence the result.
\end{proof}

We are now in a position to prove the theorem.

\begin{proof}{\textit{of theorem \ref{thm:dynamical coevents are quadratic}}}
\\By assumption our probability measure $\P_S$ obeys equation \ref{eq:probs on coevents}, then the sum rule obeyed by the quantum measure (equation \ref{eq:quantum measure sum rule}) gives us:
\bea
\P_S(\{\p|\p(A\sqcup B\sqcup C)=1\}) &=& \P_S(\{\p|\p(A\sqcup B)=1\})+\P_S(\{\p|\p(B\sqcup C)=1\}) \nonumber \\
&& +\P_S(\{\p|\p(C\sqcup A)=1\})-\P_S(\{\p|\p(A)=1\})\nonumber \\
&&-\P_S(\{\p|\p(B)=1\})-\P_S(\{\p|\p(C)=1\}), \label{eq:coevent prob sum rule}
\eea
for all disjoint $A,B,C\in\EA$. Now because $\P_S$ is a probability measure we can use the Kolmogorov sum rule to decompose its valuation on any set into a sum of its valuations on the elements of that set. Thus if we define $p_\p=\P_S(\{\p\})$ we get:
\bea
\P_S(\{\p|\p(A)=1\}) &=& \sum_{\p(A)=1} p_\p \nonumber \\
&=& \sum_{\p\in S} p_\p \tilde{\p}(A), \label{eq:coevent prob decomposition}
\eea
for all $A\in\EA$. Putting this back into equation \ref{eq:coevent prob sum rule} we get:
\bea
\sum_{\p\in S} p_\p \tilde{\p}(A\sqcup B\sqcup C) &=& \sum_{\p\in S} p_\p \tilde{\p}(A\sqcup B)+\sum_{\p\in S} p_\p \tilde{\p}(B\sqcup C)+\sum_{\p\in S} p_\p \tilde{\p}(C\sqcup A) \nonumber \\
&& -\sum_{\p\in S} p_\p \tilde{\p}(A)-\sum_{\p\in S} p_\p \tilde{\p}(B)-\sum_{\p\in S} p_\p \tilde{\p}(C),
\eea
for all disjoint $A,B,C\in\EA$, which we can rewrite as:
\beq\label{eq:decomposed coevent prob sum rule}
0 = \sum_{\p\in S}p_\p R_{ABC}(\p)~~~~\forall~\text{disjoint}~A,B,C\in\EA.
\eeq
Now fix $A,B,C$ and let $Z_{ABC}\subset S$ be the set of co-events $\p$ such that $R_{ABC}=0$. Further denote by $\overline{Z}_{ABC}$ the set of co-events such that $R_{ABC}\neq 0$, so that $S=Z_{ABC}\sqcup\overline{Z}_{ABC}$. Now by assumption every $\p\in S$ is multiplicative, so by lemma \ref{lemma:dynamical multiplicative coevents R is 0 or 1} $\p\in\overline{Z}_{ABC}\Rightarrow R_{ABC}(\p)=1$. Then equation \ref{eq:decomposed coevent prob sum rule} gives us:
\bea
0 &=& \sum_{\p\in Z_{ABC}}p_\p R_{ABC}(\p) + \sum_{\p\in \overline{Z}_{ABC}}p_\p R_{ABC}(\p), \nonumber \\
&=& \sum_{\p\in \overline{Z}_{ABC}}p_\p.
\eea
Noting that $p_\p\geq 0$, this means $p_\p=0$ for all $\p\in\overline{Z}_{ABC}$. Since $A,B,C$ are arbitrary disjoint sets, we conclude that $p_\p=0$ unless $R_{ABC}(\p)=0$ for all disjoint $A,B,C$. But by lemma \ref{lemma:dynamical coevents R and Q} $R_{ABC}(\p)=0$ implies that $Q_{ABC}(\p)=0$, which by lemma \ref{lemma:quantum sum rule implies quadratic rule} implies that $\p$ is quadratic. This completes our proof.
\end{proof}

\section{Approximate Preclusion}\label{sec:approximate preclsuion}

Rather than tackling the whole of the dynamics, an alternate approach is to focus directly on the predictive content of a histories theory. Thus instead of, or as a step towards, expressing $\mu$ in terms of the theory's allowed co-events, we would primarily be concerned with the probability statements that are the (experimentally falsifiable) predictions of our theory. As before, we seek to be able to `start' with a set of potentially real co-events, and to make our experimental predictions in terms of these co-events and structures defined in terms of them.

Our strategy will be to examine more closely the meaning (or interpretation) of probability \emph{as applied to (experimentally) falsifiable predictions}; thus we look to understand a dynamical statement $\mu(A)=q$ or $\P(B)=p$ through its (experimentally) falsifiable implications. We will begin by focusing on classical theories, where we can draw from a long history of thought upon the matter to pick a suitable interpretation of probability; we then seek to generalise this interpretation in terms of the potentially real co-events in such a manner that can be extended to quantum histories theories.

This strategy requires a somewhat rigorous approach to the interpretation of probability, for we are not simply attempting to justify pre-existing statistical practice, but to construct a predictive framework for a new ontology. In a world of co-events, what do we mean by $\mu(A)=q$, or even $\P(A)=p$? Since co-events have been connected to the measure through the preclusion of null sets, a natural approach is to seek an interpretation of probability based on this notion. We begin with a simple example that will guide us through this process.

\subsection{The Classical Coin}\label{sec:classical coin}

Consider a coin described by classical stochastic dynamics. We assume that if we throw the coin we will either have a `heads' or a `tails' result, we will further assume that all such throws are in some sense `equivalent' so that the probability of `heads' is always $p$ and the probability of `tails' always $1-p$. But what exactly do we mean by this?

One approach to its interpretation is to emphasise the role of probability in representing our knowledge or expectations about a system, so in this case we would for example say that $p=1/2$ if we had no information or view on the outcome of a coin toss \cite{Ramsey:1926,Finetti:1930,Cox:2001}. An alternate approach focuses on the practical testing of such assertions, which takes place through a set of repeated, independent trials \cite{Friedman:1999,Kolmogorov:1933}. In the latter framework a probability assertion $\P(heads)=1/2$ is a statement, not about a single trial, but about a theoretical ensemble of trials, or a prediction regarding a hypothesised set of future trials. Although the former approach is valid and useful, for example in the theory of decision making \cite{Anand:1993}, in light of our difficulties in constructing a probability measure on the space of co-events (section \ref{sec:dynamics on coevents}) we will focus on the latter approach, which is more in accordance with the predictive and experimental nature of our field, and as we shall see lends itself naturally to the concept of preclusion. We begin by discussing how we might test an assertion $\P(heads)=p$ in practice.

We test a theory, or a statement within a theory, by testing its experimentally falsifiable predictions. But when we believe there is one realised outcome (the actually real co-event, or the real history if we using the naive interpretation) we face the problem that unless $p\in\{0,1\}$, the statement $\P(heads)=p$ can not be falsified by either outcome `heads' or `tails'. We can however make falsifiable statements about sequences of trials.

The frequentist approach, as articulated by Friedman \cite{Friedman:1999}, interprets a probability statement $\P(heads)=p$ by identifying it with the asymptotic relative frequency of the outcome `heads' in an infinite sequence of repeated trials; under the requirement that the sequence of trials conforms to a `randomness' criterion \cite{Friedman:1999,VonMises:1957}. Indeed, if we had a theoretical ensemble of infinitely many trials, the proportion of the trials resulting in heads would be exactly\footnote{More precisely it would be $p$ with probability $1$.} $p$; alternatively the assertion that `the proportion of heads is not $p$' is precluded by the measure on the ensemble, and herein lies our link with co-events.

However asymptotic properties of a sequence cannot be determined using a finite number of elements of the sequence, and an infinite ensemble of trials is not realisable in terms of an experiment, whereas we seek to phrase probability in terms of experimentally falsifiable predictions. Further, the co-event is meant to be `real', so it is not meaningful to say it precludes an imagined event in a theoretical ensemble that is not part of the `actual' event algebra on which the co-event is defined.

We are thus pushed into defining probabilities based upon finite trials, and will focus on one particular technique of testing probability assertions through a sequence of repeated trials; the standard statistical technique of hypothesis testing \cite{Tucker:1962}. We present a simple version of hypothesis testing that will suffice for the purposes of this paper.

We begin with our assertion, or \emph{hypothesis}, $\P(heads)=p$ and repeat our trial $n$ times, resulting in a history $\g$ consisting of $n$ ordered `heads' or `tails' outcomes (we will ignore the issue of the `randomness' of this sequence for now). We will denote by $H$ (or $H(\g)$) the number of these trials that result in a `heads' outcome, so that the proportion of heads is $H/n$; to compare this `statistic' \cite{Tucker:1962} with the implications of our hypothesis we assume the product measure (which we shall also denote by $\P$) on the sequence of trials. Using $N_H$ to denote the event that the number of heads in the realised history is $H$, we define the \emph{cumulative probability} of $N_H$ to be:
\beq\label{eq:cumulative probability}
\CP(N_H) = \sum_{m=1}^{H}\P(N_m).
\eeq
Thus $\CP(N_k)$ is the probability that the number of heads in the realised history is less than or equal to $k$. For the purposes of this paper we will say that we \emph{reject the hypothesis at the $\e$ level} if $\CP(N_{H(\g)}) < \e$, where $0<\e\leq 1$ (typically $\e<<1$). If $\g$ is such that $\CP(N_{H(\g)}) \geq \e$ we \emph{fail to reject the hypothesis at the $\e$ level}\footnote{This is essentially a `one-tailed test'; while both one- and two-tailed tests would be equally valid in this context (though they would correspond to different values of $\e$), we will find the one-tailed test simpler to analyse.}. Note crucially that we have implicitly \textit{chosen} to organise the potential outcomes according to the number (or equivalently the proportion) of heads outcomes, based on our hypothesis; the importance of this choice will become clear below.

This technique (or a more sophisticated version thereof \cite{Tucker:1962}) can be used in the testing of scientific theories; $\e$ is then chosen to represent the `degree of certainty' we wish to test our theory to. However, assuming $p\not\in\{0,1\}$ no history (no sequence of outcomes) is actually precluded by our hypothesis; strictly speaking the hypothesis has no falsifiable implications for the realised outcomes, thus no realised history can falsify the hypothesis.

To justify our hypothesis testing technique, and to place it on a more rigorous footing, we turn to \emph{Cournot's Principle} \cite{Cournot:1843,Shafer:2005}; which has long been used to connect ``the mathematical formalism of probability to the empirical world'' \cite{Galvan:2008}.

\subsection{Cournot's Principle}\label{sec:cournot's principle}

Crudely speaking, in one way or another Cournot's Principle `rules out' events of small probability. This concept was certainly known to Bernoulli, who asserted that high probability can be treated as a `moral certainty' \cite{Bernoulli:1713}. Cournot himself made the connection with physics, arguing that events of small probability may be mathematically possible but `physically impossible' \cite{Cournot:1843}. Among many other references, the principle is used or alluded to by Levy \cite{Levy:1925,Levy:1937}, Markov \cite{Shafer:2005}, Borel \cite{Borel:1909} \& Kolmogorov \cite{Kolmogorov:1933,KolmogorovEnglish:1933}, under the name `Principle B'. More recently it has been applied by Goldstein et al to statistical mechanics \cite{Goldstein:2001} and Bohmian mechanics \cite{Durr:1992}, and by Galvan to quantum mechanics \cite{Galvan:2008}. For a more detailed discussion of the history of Cournot's principle see \cite{Shafer:2005}.

However it was the French mathematician Maurice Frechet who coined the term ``principe de Cournot'', which has come into English as `Cournot's Principle', or the `Cournot Principle' \cite{Shafer:2005}. Frechet distinguished between Strong and Weak forms of Cournot's Principle \cite{Frechet:1949}; Shafer describes these two formulations as follows \cite{Shafer:2005}:
\begin{quote}
The strong form refers to an event of small or zero probability that we single out in advance of a single trial: it says the event will not happen on that trial. The weak form says that an event with very small probability will happen very rarely in repeated trials. Some authors, including Levy, Borel, and Kolmogorov, adopted the strong principle. Others, including Chuprov and Frechet himself, preferred the weak principle.
\end{quote}
Kolmogorov's statement of Strong Cournot (his `Principle B) as concerns the probability $\P(E)$ of an event $E$ occurring in an experiment $C$ is \cite{Kolmogorov:1933,KolmogorovEnglish:1933} (as translated by Sahfer \& Vovk \cite{Shafer:2005vv}):
\begin{quote}
If $\P(E)$ is very small, one can be practically certain that when $C$ is carried out only once, the event $E$ will not occur at all.
\end{quote}
These statements contain some ambiguities, for example, in Kolmogorov's `Principle B', what exactly do we mean by `practically certain'? To apply Cournot's Principle to co-events, we will have to be more rigorous. We begin by giving a more precise specification of the Strong Cournot Principle:
\begin{description}
  \item[Literal Strong Cournot] Events of probability less than $\e<<1$ do not occur, for some $\e> 0$.
\end{description}
Where there is no danger of confusion with other variations of the Strong Cournot Principle, we will refer to Literal Strong Cournot simply as the `Strong Cournot Principle', or `Strong Cournot'. This literal specification of the Strong Cournot Principle is in some ways a `maximal' interpretation of the Cournot Principle, and as such may appear to be `neat' or `simple'; it would also allow us to deal with probability `objectively', and in the absence of observation, by grounding it in the ontology. However it is problematic in the case of classical stochastic theories. Firstly what do we mean by a small probability, are we to take $\e$ as a constant of nature? Secondly there is nothing to prevent the probability of every single history in the repeated trial being much smaller than the probability of the events we wish to preclude, in which case a literal application of Strong Cournot would rule out all single histories, and thus any reality. In fact an example of such a situation is provided by our coin, given a sufficiently large number of trials.

Since Literal Strong Cournot is not appropriate for classical theories, we instead present a `minimal' version of Weak Cournot that allows us to perform the hypothesis testing described above, but little else.
\begin{description}
  \item[Operational Weak Cournot] In a repeated trial, an event of probability less than $\e<<1$, identified ahead of the repeated trial, will not occur, for some $\e> 0$. Only a single event can be considered for a given repeated trial.
\end{description}
Thus we can only make one prediction for each repeated trial. Where there is no danger of confusion with other variations of the Weak Cournot Principle, we will refer to Operational Weak Cournot simply as the `Weak Cournot Principle', or `Weak Cournot'. Notice that our reference to repeated trials is essentially semantic; a single trial can be considered as a special, `$n=1$', case of a repeated trial, and a general `$n$ times' repeated trial can be considered as a single trial in which sequences of length $n$ are thought of as single outcomes. We insist on the term `repeated trial' to emphasise the link between Operational Weak Cournot and hypothesis testing\footnote{Indeed, we regard Operational Weak Cournot as a variation on Shafer's `weak principle' precisely because of the shared emphasis on repeated trials. We note however, that in its insistence that certain `low probability' events \emph{will not} occur (rather than \emph{occurring rarely}, Operational Weak Cournot bears some resemblance to Shafer's strong principle'.}.

Operational Weak Cournot has the advantage that it fits in well with our actual methods of falsifying theories, such as the statistical hypothesis testing discussed in section \ref{sec:classical coin}. When we make the hypothesis $\P(heads)=p$, there may be certain outcomes that would convince us that this hypothesis has been falsified; for example we may reconsider the assertion $\P(heads)=1/2$ were we to toss our coin a million times only to find a heads outcome for every toss. We can turn this around by saying that \textit{those outcomes that would falsify a theory are precluded by it}. Of course there remains some ambiguity in the choice of $\e$, however this ambiguity is inherent and has not been introduced by our adoption of Weak Cournot; whatever interpretation of probability we adopt we would have to choose the $\e$ we use to falsify our theories (of course this $\e$ need not be unique). In a sense this identification takes the arbitrariness of our choice of $\e$ out of our theory and into the `meta' level on which we compare and reject theories. The `meta' level is always present and by using it to give us $\e$ we have avoided the addition of `new' ambiguity.

Applying this to our coin (section \ref{sec:classical coin}), we start with our hypothesis $\P(heads)=p$, then in the context of a `possible' experiment (or $n$-fold repeated trial) we \emph{choose} to consider the proportion of heads that would be realised in such as experiment, and thus single out the event\footnote{This is an event in the event algebra related to the sample space of all possible sequences of outcomes in our repeated trial.} $L$ that the realised history $\g$ will obey $\CP(H(\g))<\e$. Since the probability of this event is less than $\e$, using Operational Weak Cournot we are justified in asserting that it will not occur (we are assuming the hypothesis). The non-occurrence of $L$ then becomes a falsifiable prediction of our hypothesis; and once we perform the experiment, the occurrence of the event $L$ would falsify the hypothesis.

However, given the realised history $\g$ (once the experiment has been carried out), if $n$ is sufficiently large we will always be able to find an event that has occurred (ie it contains the realised history) yet is assigned probability $<\e$ by the hypothesis; for example, as mentioned above the event that the realised history $\g$ will occur will itself be of probability $<\e$ for sufficiently large $n$. Because of this, to rule out the problems faced by Literal Strong Cournot we have phrased Operational Weak Cournot so as to allow conclusions to be drawn concerning only the single\footnote{Each such prediction is an element of the event algebra related to the sample space of all possible sequences of outcomes in our repeated trial. thus if we wish to make multiple predictions in the context of a single experiment, we can combine them using the logical (Boolean) operations in the event algebra to yield a single event.} falsifiable prediction (implied by our hypothesis) singled out before the experiment and being tested by it.

In this way Operational Weak Cournot provides us with a practical method of relating probability statements to experimental measurements, though we have had to adopt a `minimalist' and perhaps `restrictive' interpretation of probability, which might for example restrict our ability to treat probability statements as logical predicates \cite{Cox:2001}. However in the classical theory this has no impact on either our ontology or the `meaning' we attribute to dynamics, for we typically assume a deeper deterministic theory operating at a more fundamental level. This will not hold in the quantum case, which therefore will require us to take our chosen interpretation of probability more seriously.

\subsection{Approximate Co-Events}\label{sec:approximate co-events}

In one way or another Cournot's principle rules out events of small probability. Because this is a straightforward generalisation of the concept of preclusion, which rules out null sets, we can easily express it in terms of co-events by introducing the concept of \textit{approximate preclusion}. For the sake of clarity we will henceforth refer to preclusion itself as \textit{exact preclusion}, thus an exactly preclusive co-event (which we will sometimes shorten to an \emph{exact co-event}) obeys
$$\mu(A) = 0 \Rightarrow \p(A)=0$$
whilst an \textit{approximately preclusive co-event} is given by:
\begin{definition}\label{def:approximate preclusion}
Let $\H$ be a histories theory. Given $\e>0$ a co-event $\pe$ is \textbf{approximately preclusive} at the $\e$ level if
$$\mu(A)<\e \Rightarrow \pe(A)=0.$$
We say that $\pe$ is an \textbf{approximately preclusive co-event}, or simply an \textbf{approximate co-event}.
\end{definition}
Notice that this definition holds for a general histories theory, not simply in the classical case, and thus might be used to generalise Cournot's Principle.

Though the introduction of approximate co-events is a fundamental change in the the theory, its implications for the `internal structure' of the `allowed' co-events is less than radical. To be precise, we have effectively altered the precluded events, including `almost null' as well as null sets. We will refer to events of measure less then $\e$ as `approximately precluded' or `$\e$-null' sets, with `$\e$-negligible' sets defined similarly. However other than this we have made no changes, and have not altered the algebraic structure of the co-events, and so can define many of the same concepts and prove many of the same theorems that we could for exactly preclusive co-events. In particular we can define multiplicativity, domination and primitivity based upon approximate preclusion in exactly the same way we defined these concepts for exact preclusion, leading to a notion of `approximate co-event schemes' and in particular an approximate multiplicative scheme, which we shall henceforth assume. Further, we can prove `approximate co-event versions' of many of the results we have established for the exact multiplicative scheme, for example the following lemma, which shall make use of below, is analogous to lemma \ref{lemma:mult existence primitivity}:

\begin{lemma}\label{lemma:approximate mult existence primitivity}
Let $\H$ be a history theory with a finite sample space and let $A\in\EA$ be non-negligible (ie $A$ is not a subset of a null set). Then there exists a primitive approximately preclusive multiplicative co-event $\pe$ such that $\pe(A)=1$.
\end{lemma}
\begin{proof}
See proof of lemma \ref{lemma:mult existence primitivity}.
\end{proof}

In a classical theory, how we interpret $\p_\e$ depends on the version of Cournot's Principle we are following. Weak Cournot means that the $\p_\e$ are essentially theoretical tools, used to phrase (experimentally falsifiable) predictions in terms of preclusion. In this case $\e$ will be taken from our `meta' level choice of $\e$ used to falsify a given theory, and may not be the same for every system. On the other hand Strong Cournot means that the $\p_\e$ are the potential realities, in which case we may consider $\e$ as a constant of nature. Note that our basic ontology remains unchanged, the `actually real' co-event remains our description of reality and its internal structure (or logic) is still given by the multiplicative rule\footnote{We assume that the multiplicative scheme holds for all histories theories, in a classical theory it `happens' to coincide with the classical scheme.}. What we have done is to alter the set of potential realities \textit{given} a measure, thus essentially we have altered the role \& meaning of the measure.

Moving from classical stochastic theories to quantum mechanics and the multiplicative scheme, we are no longer treating a single history as real\footnote{More precisely we are no longer expecting our primitive multiplicative co-events to be classical} so our objections to Strong Cournot may no longer be relevant. This leads us to question whether it may be possible to achieve Strong Cournot in the context of quantum measure theory, and to take the $\pe$ literally. One advantage of such an approach lies in our practise of treating some events of small probability as null. In section \ref{sec:goals} we raised concerns regarding the preclusion of `approximately null' sets, arguing that true null sets are rare; in fact due to experimental inaccuracies the events we are characterising as null will in general be only of small probability. Approximate preclusion is tailor made to address such concerns, in particular if we are able to adopt Strong Cournot we can preclude such sets directly.

If Strong Cournot, as expressed by approximate co-event, is to be applicable to a general quantum measure theory, it must certainly make sense in the context of classical theories. In the next section we therefore use approximate preclusion to explore the application of Strong Cournot to hypothesis testing in a classical theory; returning to the example of our simple coin (section \ref{sec:classical coin}). In what follows we assume that all co-events are multiplicative.

\subsection{Can we Achieve Strong Cournot?}\label{sec:strong cournot}

\subsubsection{Approximate Preclusion for a Coin}

As before we assume that if we throw the coin we will either have a `heads' (`$h$') or a `tails' (`$t$') result, we will further assume that all such throws are in some sense `equivalent' so that the probability of `heads' is always $p$ and the probability of tails always $1-p$. Thus in the histories formalism, a single throw corresponds to the sample space: \beq\label{eq:coin sample space}
\O=\{h,t\}.
\eeq
In the language of our hypothesis testing technique (section \ref{sec:classical coin}), our hypothesis is that $\P(\{h\})=p$, so that the classical measure $\P$ is:
\bea
\P(\{h\})&=&p,\nonumber \\
\P(\{t\})&=&1-p,\nonumber \\
\P(\{h,t\})&=&1.\label{eq:coin dynamics}
\eea
An `experiment' consisting of $n$ trials (which we assume to be independent) has the sample space
\beq
\O = \{h,t\}^n,
\eeq
of histories (ordered sequences of outcomes) $\gamma=a_1,\ldots , a_n$, where $a_i\in\{h,t\}$. The corresponding event algebra is $P\O$ and since we have assumed that our trials are independent we use the product measure, which by abuse of notation we shall also denote by $\P$.

Now we have previously denoted the number of `heads' outcomes in a history $\g$ by $H(\g)$, so that the proportion of `heads' is $H(\g)/n$. Then we have:
\beq\label{eq:coin single history probability}
\P(\{\g\})=p^{H(\g)}(1-p)^{n-H(\g)}.
\eeq
We are less interested in the individual histories than in the proportion of heads, which, as discussed above, should reflect the underlying probability $p$ of a heads outcome in a single trial. We will denote by $N_H$\ the event that the number of heads realised in the sequence is $H$. Then we have:
\beq\label{eq:coin probability of H heads}
\P(N_H) = \left( \begin{array}{c} n \\ H \\ \end{array} \right) p^{H}(1-p)^{n-H}.
\eeq
Further, we will label by $L_H$ and $G_H$ the events that the number of heads in the realised sequence is less than or equal to or greater than $H$ respectively. Then we get:
\bea
\P(L_H)&=&\sum_{m\leq H}\P(m), \nonumber \\
&=& \CP(N_H). \label{eq:coin proabaility of H or less heads}
\eea
Thus if we performed such an experiment and realised a proportion of heads corresponding to a small $\P(L_H)<\e$ we would reject our hypothesis $\P(h)=p$ (at the $\e$ level). Turning this around, given the assumption $\P(h)=p$ we wish to preclude all events $L_H$ (and $N_H$) with $\P(L_H)<\e$ for some small $\e$.

We can make this more precise; since the measure $\P$ is classical and non-zero everywhere, and since $L_m$ is a proper subset of $L_{m+1}$, we can see that $\P(L_m)$ is monotonic in $m$. Thus defining $H_\e$ as the greatest $H$ such that $\P(L_H)<\e$, we have $\P(L_{H_\e+m})\geq\e ~\forall ~m>0$; in particular $\P(L_{H_\e})<\e\leq\P(L_{H_\e+1})$. Thus, following the argument of section \ref{sec:classical coin} we would like to say that $L_{H_\e}$ is precluded while $L_{H_\e+1}$ is not. Since $\P(L_H)>0$ in all cases, we turn to approximate preclusion.

Adopting Strong Cournot, we fix $\e$ and assume that events of measure less than $\e$ never occur. Then we can immediately rule out $L_{H_\e}$, indeed $\p_\e(L_{H_\e})=0$ for all approximately preclusive co-events $\p_\e$. However, this does not hold for $L_{H_\e+1}$, indeed since $\P(L_{H_\e+1})>\e$ by lemma \ref{lemma:approximate mult existence primitivity} there exists a primitive approximate co-event $\p_{H_\e}$ mapping it to unity. In other words:
\bea
\p_{H_\e}(L_{H_\e})&=&0, \nonumber \\
\p_{H_\e}(L_{H_\e+1})&=&1.
\eea
Thus every potential reality precludes $L_{H_\e}$, but not every potential reality precludes $L_{H_\e+1}$. This is enough to allow us to use the hypothesis testing technique outlined in section \ref{sec:classical coin}, thus enabling us to deal with (experimentally falsifiable) predictions in terms of our approximate co-events. Since approximate preclusion is defined in terms of a general histories theory, this seems to raise the possibility of a predictive approximate co-event framework for a general quantum measure theory. However, though widening the scope of preclusion by moving from exact to approximate co-events has introduced useful features to our theory, we must make sure that it has not also introduced problems.

\subsubsection{The Failure of Strong Cournot}

The above construction seems promising, it seems that our shift from exact to approximate preclusion has succeeded in introducing new \& useful features to co-event theory; however we must check that it has not inadvertently introduced new problems as well. In dealing with exactly preclusive co-events we can show that (under certain conditions) classical measures always imply classical outcomes meaning that the co-events will necessarily behave classically (see section \ref{sec:the multiplicative scheme}); do the approximate co-events always `behave'? Such concerns lead to three objections to the construction we outlined above:

\begin{description}

\item[I The Status of Single Histories]~

    The first question to be raised regards the status of single histories. In a system obeying non-classical dynamics we may be happy with single histories not being realised, however in this system with its classical dynamics we could still find all single histories to be ruled out for a sufficiently large number of trials. This raises interpretational issues, for example if we take $p=1/2$, $\e=10^{-3}$ and $n=10$ we find all single histories ruled out. However if $n=9$ all single histories are allowed. We have several problems here, firstly we may not be comfortable with the non-classical behaviour of the $n=10$ system. Secondly there is the question of the value of $\e$, which leads to classical outcomes in the $n=9$ case but not the $n=10$ case. Finally it seems odd that one additional trial will `disallow' classical behaviour, particularly given that a single trial behaves classically in and of itself.

    Note that we don't have a causality problem here; although in the $n=10$ system no single history is realised this will not become apparent in the first trial. The two outcomes of the first trial can be thought of as a coarse grained partition of the $n=10$ sample space, and as such each outcome has probability greater than $\e$. This will also hold in general, any single history in the $n=10-m$ system (where $1\leq m\leq 9$) corresponds to a coarse grained event in the $n=10$ system that has probability greater than $\e$ and so is not precluded.

    Further note that this problem of systems that behave classically becoming non-classical in a repeated trial can occur in the histories framework without the use of co-events. We could take the view that no system is truly classical, but rather that at the level of observable events we have an emergent classicality based on environmental decoherence. Adopting this view, observed `classical' behaviour is due to approximate decoherence \cite{Sorkin:private,Dowker:private}, and repeated trials of an approximately decohering system (such as a `quantum coin' with interference $\varepsilon$ between heads and tails) may lead to non-classical outcomes, notably the reappearance of quantum coherence \cite{Sorkin:private}.

\item[II The Problems of Assigning Ontology to the Approximate Co-Events]~

    The second question to be raised concerns the range of possible co-events `allowed' by approximate preclusion. We may be able to construct allowed co-events $\phi_\e$ such that some observable questions have no definite answer. In other words our framework allows potential realities that if realised would imply observable anhomomorphisms.

    For example, $\{L_{H_\e},G_{H_\e}\}$ is a partition of $\O$, both  elements of which are observable (since all events are observable in this classical system). Then using lemma \ref{lemma:approximate mult existence primitivity} there will be at least one primitive approximate co-event $\p_{G_\e}$ mapping $G_{H_\e}$ to $1$, and $\p_{G_\e}$ will map $L_{H_\e}$ to $0$ (as do all approximate co-events) because $\P(L_{H_\e})<\e$. However, although we can `cherry-pick' the co-event $\p_{G_\e}$ to treat the partition $\{G_{H_\e},~L_{H_\e}\}$ classically, we do not know if $\p_{H_\e}$ will map $G_\e$ to one and so treat this partition classically. More generally we do not have any guarantee that all the allowed co-events will treat this partition classically; in fact we can find primitive approximate co-events that map both elements of the observable partition $\{G_{H_\e},L_{H_\e}\}$ to zero.

    We can find subsets of histories $T\subset L_{H_\e}$ and $P\subset G_{H_\e}$ such that $\P(T),\P(P)<\e$ but $\P(T\sqcup P)>\e$.  Now $\P(T\sqcup P)>\e$ means that by lemma \ref{lemma:approximate mult existence primitivity} some subset $C$ of $T\sqcup P$ will be the dual of a primitive approximate co-event $C^*$. Further both $T$ and $P$ are approximately precluded so $C$ cannot be a subset of either; therefore $C$ has non-empty intersection with both $L_{H_\e}$ and $G_{H_\e}$, which means:
    \bea
    C^*(L_{H_\e})&=&0 \nonumber \\
    C^*(G_{H_\e})&=&0.
    \eea
    For an explicit example consider the case $p=1/2,~n=10^3$ and $\e=10^{-3}$; then $H_\e=450$ since $\P(L_{450})<\e$ whereas $\P(L_{451})\geq\e$. Now $L_{451}=L_{450}\sqcup N_{451}$, so there is some (not necessarily unique) subset $S\subseteq N_{451}$ such that $\P(S\sqcup L_{450})\geq\e$ whereas $\P((S\setminus\{\g\})\sqcup L_{450})<\e$ for any $\g\in S$. We can think of constructing $S$ by adding fine grained histories from $N_{451}$ to $L_{450}$ one by one until the measure is greater than $\e$; thus we can not reduce the set $S\sqcup L_{450}$ without its measure falling below $\e$. In fact, because $p=1/2$ (and the measure is classical) every single history $\g$ contributes the same amount, $p^n=2^{-10^3}$, to the probability of any event containing $\g$, so that $\P(S)=p^n |S|$. Further, since the measure is classical we know that $\e\leq\P(S)+\P(L_{450})<\e+p^n$. Therefore: \beq
    p^{-n}(\e - \P(L_{450})) \leq |S| < p^{-n}(\e - \P(L_{450})) + 1,
    \eeq
    so that $S$ could be any subset of $N_{451}$ of cardinality:
    \bea
    |S|&=&Int(\frac{\e-\P(L_{450})}{p^n})\nonumber \\
    &\approx& 1.4\times10^{297},
    \eea
    where $Int(x)$ denotes the least integer which is greater than or equal to $x$. Then the set $S\sqcup L_{450}$ has measure greater than or equal to $\e$, and so is not a subset of any $\e$-null set (since the measure is classical). Further, we cannot reduce our set without the remainder being $\e$-negligible. Therefore $S\sqcup L_{450}$ is the base of an approximate co-event $(S\sqcup L_{450})^*$ that maps both $G_{H_\e}$ and $L_{H_\e}$ to zero.

\item[III The Inconsistency of Multiple Partitions]~

    Finally, in the examples we have looked at so far, we considered a single partition $\O=G_{H_\e}\sqcup L_{H_\e}$ which was treated non-classically by some allowed co-events. Intuitively, both elements of this partition correspond to `meaningful' propositions; that the proportion of heads in the observed history is less than or equal to ($L_{H_\e}$), or greater than ($G_{H_\e}$), $H_\e/n$. However, we did not attach any particular `meaning' to our `problem co-event' $C^*$, indeed as a proposition the related event $C$ is `pathological' in that it would be difficult to express it as an `English sentence' as we were able to do for $L_{H_\e}$ and $G_{H_\e}$ above. This might lead one to speculate that we could perhaps focus on partitions and co-events corresponding to `good', `meaningful' or `useful' propositions, whilst ignoring partitions and co-events corresponding to `bad' or `pathological' propositions; indeed we were able to find the `good' approximate co-event $\p_{G_\e}$ that was classical on the `good' partition $\O=G_{H_\e}\sqcup L_{H_\e}$.

    However the situation is not so simple; in the above we focused on a single partition, whereas there are many `good' partitions, each with its associated `good' co-events. Unfortunately the co-events that treat one partition classically will not in general be classical on another partition, as the following example shows.

    When the number of trials of our coin is even, $n=2m$, we can partition the sample space $\O$ into even and odd coarse grainings as follows. Given a single history $\g=a_1,a_2,\ldots,a_{2m}$ (where $a_i\in\{h,t\}$) we can form the \textit{even} and \textit{odd} histories $E(\g)=a_2,a_4,a_6,\ldots,a_{2m}$ and $O(\g)=a_1,a_3,a_5,\ldots,a_{2m-1}$. Likewise we can form the even and odd coarse grainings $\O_E=\{E(\g)|\g\in\O\}$ and $\O_O=\{O(\g)|\g\in\O\}$, which inherit their associated measures from $\O$, and it is easy to see that $\O=\O_E\sqcup\O_O$. Building on this, we can treat both $\O_E$ and $\O_O$ in the same way that we previously treated $\O$ by looking at the number of heads in the even and odd trials, $H^E$ and $H^O$ respectively. Given an $\e$ we can go on to define $H^E_\e,H^O_\e$, and the subsets $G_{H^E_\e},G_{H^O_\e}$ and $L_{H^E_\e},L_{H^O_\e}$ of the two distributions. Then, following the analysis above, we can define the co-events $\p^E_{G^E_\e},~\p^O_{G^O_\e}$  corresponding to the questions `is the number of heads in the even distribution greater than or equal to $H^E_\e$?' and `is the number of heads in the odd distribution greater than or equal to $H^O_\e$?'. Now as before we have
    \bea
    \p^E_{G^E_\e}(L_{H^E_\e})&=&0 \nonumber \\
    \p^E_{G^E_\e}(G_{H^E_\e})&=&1, \label{eq:approx coevent multiple partitions classical}
    \eea
    however now we also have
    \bea
    \p^E_{G^E_\e}(L_{H^O_\e})&=&0 \nonumber \\
    \p^E_{G^E_\e}(G_{H^O_\e})&=&0, \label{eq:approx coevent multiple partitions nonclassical}
    \eea
    with $\p^O_{G^O_\e}$ showing similar behaviour.

    For an explicit example consider the case $p=1/2$, $m=10^3$ and $\e=10^{-3}$; thus the even and odd distributions are similar to the example considered above. Then $H^E_\e=H^O_\e=450$, and as before every single history $\g$ contributes the same amount, $p^{2m}=2^{-2*10^3}$, to the probability of any event containing $\g$; thus the measure of any event $S$ is given by $\P(S)=p^{2m} |S|$. From this we can see that an event $S$ is a subset of an approximately precluded set $T$ if and only if $S$ itself is approximately precluded, thus an approximate co-event $\p_\e$ is preclusive if and only if its dual $\p_\e^*$ has measure greater than or equal to $\e$; then if $\p_\e$ is preclusive
    \beq\nonumber
    \e \leq \P(\p_\e^*) = p^{2m}|\p_\e^*|,
    \eeq
    so that
    \beq\nonumber
    |\p_\e^*| \geq \e p^{-2m} = 10^{-3} \times 2^{2\times 10^3}.
    \eeq
    Thus $\p_\e$ is primitive if and only if $|\p_\e^*|=Int(\e p^{-2m})$, where $Int(x)$ denotes the least integer that is greater than or equal to $x$. Now let $\g_E=a_1,a_2,\ldots,a_{2m}$ be the history defined by the outcomes
    \beq\nonumber
    a_i = \left\{\ba{cc} h & i~even \\ t & i~odd. \ea\right.
    \eeq
    Then $E(\g_E) = hhhh\ldots$ and $O(\g_E) = tttt\ldots$, so $\g_E\in L_{H^O_\e}\cap G_{H^E_\e}$. Now it is easy to see that $\P(G_{H^E_\e})>\e$, so that $|G_{H^E_\e}|>\e p^{-2m}$ and thus we can find an event $C\subset G_{H^E_\e}$ of cardinality $Int(\e p^{-2m})$ that contains the history $\g_E$. But then setting $\p^E_{G^E_\e}=C^*$, we have found a primitive approximately preclusive co-event satisfying equations \ref{eq:approx coevent multiple partitions classical} \& \ref{eq:approx coevent multiple partitions nonclassical}. A similar construction can be made for $\p^O_{G^O_\e}$.

    This is in fact a coarse graining problem (as are most of the problems in the multiplicative scheme) reminiscent of the interpretational problems of the consistent histories approach. Essentially we can recover the probabilities in the measure (or decoherence functional) but we are also recovering the interpretational problems in the sense that different coarse grained partitions have become incompatible. In consistent histories each question `makes sense' in one decoherent partition, but may not have a classical answer in other decoherent partitions. In our case every approximate co-event (which corresponds in a natural way to a question via its dual) will treat at least one partition classically (for example the partition consisting of its dual and the complement thereof) but may not treat other partitions classically, even though these partitions are dynamically classical (in that they decohere). However this feature is more of an issue for approximate preclusion than for consistent histories, since, as our coin example shows, approximate preclusion may have difficulties even when the underlying fine grained histories obey classical dynamics, a problem not shared with consistent histories.
\end{description}

We can gain further insight into these problems by applying the machinery of the `principle classical partition' introduced in appendix \ref{appendix:principle classical partition}. Essentially, even though the measure is classical we are finding non-classical behaviour at the level of the approximate co-events. This might suggest that our observable partitions are finer grained that the principle classical partition, and indeed while we are assured that when using exact preclusion an classical measure will lead to a fully fine grained principle classical partition (every non-null singleton set is an element of the partition), we do not have such a guarantee when using approximate preclusion.

It is instructive, therefore, to calculate the principle classical partition for our repeated trial; to provide an explicit example we set $p=1/2,~n=10^3$ and $\e=10^{-3}$ as before. Then every history $\g\in\O$ is of equal probability $\P(\{\g\})=2^{-10^3}$, and the Kolmogorov sum rule means that every set $S\subset\O$ of trials has a probability determined solely by its cardinality, $\P(S)=2^{-10^3}|S|$. Then if we set $m=Int(2^{10^3}\e)$ (where $Int(x)$ denotes the least integer which is greater than or equal to $x$), we see that $\P(S)\geq\e$ for every set of cardinality greater than or equal to $m$ and that $\P(S)<\e$ for every set of cardinality less than $m$. Thus if $S_m$ is of cardinality $m$, it is not itself approximately precluded, however every subset of $S_m$ is approximately precluded. Thus $S_m^*$ is a primitive approximately preclusive co-event. Since $m>2$, given any two trials $\g_1,~\g_2\in\O$ we can find a primitive approximately preclusive co-event $\p_\e=S_m^*$ whose dual contains both trials, $\g_1,\g_2\in S_m$. In other words, if we fix $\g_1$, given any $\g_2\in\O$ we can find a primitive $\p_\e^{\g_2}$ whose dual contains both $\g_1$ and $\g_2$. But then the transitive closure of the duals of all these primitive approximate co-events is the whole sample space $\O$, and the principle classical partition is simply $\{\O\}$.

As a final comment on Strong Cournot, we note that the value of $\e$ used in this section may be considered `large', in that we can find experimentally observable events with smaller probability. Following this line of thought we might attempt to evade the above objections to Strong Cournot by interpreting the results outlined above as placing a `cap' on the level of $\e$ (which in Strong Cournot we are taking to be a constant of nature). Thus we would conclude that the value of $\e$ must be small compared to the probabilities of observable events (including experimentally achievable repeated trials)\footnote{Repeated trials in which $n$ is not so large that the repeated trial is experimentally unrealisable in practise}. There are however objections to this argument. Firstly, such an approach would be unlikely to address the questions surrounding the status of single histories; in particular, with regard to an anticipated theory of quantum gravity, the individual `histories of the universe' may not have `large' measure. Secondly, and more seriously, it is our use of $\e$ ro preclude observable events which allows us to discuss probability in the language of co-events, the very reason we turned to Cournot's Principle. Decreasing the value of $\e$ decreases our ability to falsify probability hypotheses, so a value of $\e$ so small as to avoid the problems outlined above would not give approximate preclusion any significant advantage over exact preclusion in the discussion of dynamics. Thus even if we adopted Strong Cournot with such a small $\e$, we would in any case be driven to some other mechanism (such as Weak Cournot) to deal with dynamics \& prediction. For these reasons, the authors prefer to discard Strong Cournot in favour of Weak Cournot.

\subsection{Quantum Operational Weak Cournot}\label{sec:quantum weak cournot}

Although Strong Cournot may be more philosophically satisfying, in particular with regard to co-event theory, the arguments of section \ref{sec:strong cournot} above conclusively demonstrate the failure of its application to multiplicative co-events due to violations of observable classicality. As in classical stochastic theories, this leads us to fall back on Weak Cournot. Though perhaps less philosophically satisfying than the strong variety, it may be the `best we can do', at least at the present time.

Our application of Operational Weak Cournot to quantum histories theories is similar in some ways to its classical application (section \ref{sec:cournot's principle}). We again evoke the `meta-level' process by which we falsify theories to give us $\e$, which consequently is no longer considered as a constant of nature and may be different for different systems. Our restriction of approximate preclusion to events singled out ahead of a repeated trial avoids the problems of multiple partitions and single histories encountered above. Essentially, we depart from Strong Cournot by introducing a split between the ontology and the predictive content of the theory; whereas in Strong Cournot both are described by a primitive approximate co-event we now propose:
\newline\newline
\textbf{Ontology}: The potential realities are primitive exactly preclusive (multiplicative) co-events. In the example of the classical coin above, since the measure is classical, these primitive co-events correspond to single fine grained histories.
\\\textbf{Predictions}: In an experiment consisting of repeated trials, an experimentally observable event $A$ singled out in advance, of measure less than $\e$, will not occur. We can of course phrase this in terms of approximate co-events, replacing the phrase `of measure less than $\e$' by `precluded by all approximate co-events $\p_\e$'. However this is tautological, and the use of approximate co-events is now rather vacuous since the ontology is given by primitive exactly preclusive co-events.
\newline\newline
This formalism aims to allow us to use the hypothesis testing technique described in section \ref{sec:classical coin} to test `Copenhagen predictions' (discussed in section \ref{sec:goals}) through experiments consisting of repeated trials, for both classical and non-classical systems. There are, however, several objections we can make.

Firstly, the restriction of our predictive ability to such experiments may obstruct the application of Weak Cournot beyond the Copenhagen framework, and the need to single out our predictions in advance of an experiment raises questions regarding observer independence; a hallmark of the histories approach. Thus even if Weak Cournot is successful in justifying the Copenhagen predictive mechanism from the perspective of quantum measure theory it is doubtful that it will assist us to develop a wider conception of the predictive content of a general quantum measure theory as hoped for in section \ref{sec:goals}. Simply by choosing to focus solely on (experimentally falsifiable) predictions we have constructed a formalism in which we cannot attach any (`physical') meaning to non-predictive dynamical statements. While such an `instrumentalist' view of dynamics may be appealing to some \cite{Dowker:private}, the authors feel that the dynamics should express something about the structure of reality, even in the absence of observation.

Had Literal Strong Cournot been applicable it would have avoided such issues, since Cournot's Principle would have been applied at the ontological level and thus its implications would have been observer independent. More importantly this split between the ontology and the predictive mechanism separates us not only from Strong Cournot but also from the application of Weak Cournot in classical physics. In a classical theory we typically assume an underlying deterministic dynamics; the dynamical statements of this `true theory' would then have observer independent meaning, allowing probability to be considered as an instrumentalist phenomena without any implication for the interpretation of this `true dynamics'.

Further, in a classical stochastic theory our statement of Weak Cournot is made in terms of the event algebra and the measure, which are both defined in terms of the histories; the potential realities of the theory. In quantum measure theory our potential realities are primitive exactly preclusive multiplicative co-events, however our formulation of Weak Cournot is still expressed in terms of events and the measure defined on them, or alternatively in terms of approximate co-events. Thus contrary to the goals we set out in section \ref{sec:goals} we have not been able to express the predictive content of quantum measure theory in terms of its potential realities.

Finally, by adopting the above formulation of Weak Cournot we have not made headway in our attempt to deal with `almost null' sets (section \ref{sec:goals}). Given our practice of treating `almost null' sets as `exactly null' by assuming idealised conditions our inability to address this issue is a cause for concern.

\section{Conclusion}\label{sec:conclusion}

Our aim in this paper has been to rephrase the dynamical \& predictive content of a quantum histories theory in terms of (multiplicative) co-events, or at the very least to demonstrate that this is possible in principle even if in practice it may be simpler to work with event algebras. In particular, we were concerned that due to experimental inaccuracies the events that we are characterising as null will in general be only of small probability.

Our first, and potentially most appealing proposal was to shift the dynamics wholesale onto the co-events, `completing' the co-event program by moving from `histories theories' to `co-event theories' in which we could use co-events and the dynamics defined on them to deal directly with physical systems. This would have allowed us to deal with events of small probability by assigning small probabilities to the co-events that found those events to be true. Unfortunately this approach has stalled.

Our second proposal was approximate preclusion. Following Strong Cournot we aimed to explain the emergence of probability from preclusion by altering the status of the measure with regard to the allowed co-events. This would have allowed us to deal with events of small probability by directly precluding them. Again, this approach has not met with success.

Finally we have been driven to adopt Weak Cournot, which can perhaps be used to justify the (existing) Copenhagen experimental framework from the vantage point of a quantum histories theory. However we do not manage to express our predictions in terms of the `potentially real' co-events themselves, and make no headway in extending our understanding of a histories theory's predictive content beyond the narrow Copenhagen framework. Further, we are left in a quandary regarding the application of preclusion in practice. As we have previously pointed out the events that we are assuming to be null, for example in the double slit experiment, are only found to have zero measure following the assumption of an idealised system that we know cannot hold in practice. Thus for our results to have application to the real world an alternate solution to this problem must be found.

\section{Acknowledgments}

The authors want to thank Fay Dowker, Rafael Sorkin and Sumati Surya for many insightful discussions on the co-event interpretation. The authors further thank Fay Dowker, Anna Gustavsson and Rafael Sorkin for their comments on drafts of this paper. This work was partly supported by the Royal Society-British Council International Joint Project 2006/R2. YGT was supported by an STFC studentship.

\appendix
\section{The Principle Classical Partition}\label{appendix:principle classical partition}

For the purposes of this appendix, where there is no scope for confusion we will sometimes abbreviate `classical with respect to $\M$' to `classical'. We can extend the notion of fine graining to a partial order on the partitions of $\O$. We say that a partition $\la$ of $\O$ is a \emph{fine graining of}, or equivalently is \emph{less than}, a partition $\Theta$ of $\O$ if $\Theta$ can be regarded as a partition of $\la$. Equivalently we could say that $\Theta$ is a \emph{coarse graining}, or is \emph{greater than}, the partition $\la$. It is easy to see that fine graining induces a partial order on the space of partitions of $\O$.

Restricting to classical partitions, we can show that this partial order contains a unique minimal element, the `finest grained' classical partition, which we call the \emph{principle classical partition} (with respect to $\M$).

\begin{theorem}\label{thm:principle classical partition}
Let $\H$ be a histories theory with a finite sample space. Then there exists a classical partition (with respect to $\mu$) $\la$ of $\O$ that is minimal among the classical partitions of $\O$. Furthermore $\la$ is unique.
\end{theorem}

To prove this we first need a technical lemma:

\begin{lemma}\label{lemma:classicality equals homomorphism}
Let $\H$ be a histories theory with a finite sample space. Then the partition $\la$ of $\O$ is classical with respect to $\mu$ if and only if for every $\p\in\M$ there exists $A_i\in\la$ such that $\p^*\subset A_i$.
\end{lemma}
\begin{proof}
First assume that $\la = \{A_i\}_{i=1}^n$ is a classical partition of $\O$. Then by assumption $\p\in\M$ will behave classically on $\la$, which means it will act as a homomorphism; obeying linearity (equation \ref{eq:linearity}) and multiplicativity (equation \ref{eq:multiplicativity}) when restricted to the subalgebra generated by $\la$. Then:
\begin{enumerate}
  \item Assume that $\p(A_i)=0~\forall i$; linearity then implies that $\p(\O)=0$. Then multiplicativity implies that $\p$ is the zero map, which is excluded from $\M$, contradicting our assumptions.
  \item Assume $\p$ maps more than one element of $\la$ to one; without loss of generality $\p(A_1)=\p(A_2)=1$. Noting that the $A_i$ are disjoint (as elements of a partition), multiplicativity implies $\p(\emptyset)=\p(A_1A_2)=\p(A_1)\p(A_2)=1$, contradicting our assumption that $\p$ is a co-event.
\end{enumerate}
Therefore there exists $j$ such that $\p(A_i)=\delta_{ij}$. But then $\p^*\subset A_j$.

To prove the converse, note that because the $A_i$ are disjoint every element of the subalgebra generated by $\la$ is of the form $B_I=\bigsqcup_{i\in I} A_i$ for some indexing set $I$. Now let $\p\in\M$, then by assumption there exists $A_k\in\la$ such that $\p^*\subset A_k$. Then for all $B_I$ in the subalgebra generated by $\la$ we have:
\beq\nonumber
\p(B_I) = \left\{\ba{cc} 1 & k\in I \\ 0 & k\not\in I. \ea\right.
\eeq
Thus it is easy to see that $\p$ is a homomorphism on the subalgebra generated by $\la$.
\end{proof}

We are now in a position to prove the theorem.

\begin{proof}{\textit{of theorem \ref{thm:principle classical partition}}}
\begin{description}
  \item[Existence:] The partition $\{\O\}$ is always classical.
  \item[Minimality:] Because $\O$ is finite the set of partitions thereof, and in particular the set of classical partitions thereof, is also finite. It is easy to see that the `fine graining of partitions' defines a partial order in the set of classical partitions of $\O$, which must therefore contain a minimal element $\la$.
  \item[Uniqueness:] We enumerate the primitive preclusive multiplicative co-events by setting $\M=\{\p_i\}_{i=1}^{m}$ where $m=|\M|$. Now let $\la_A=\{A_i\}_{i=1}^{n_A}$ and $\la_B=\{B_i\}_{i=1}^{n_B}$ be distinct partitions that are both minimal among the set of partitions of $\O$ that are classical with respect to $\M$. Then for every $\p_i\in\M$ lemma \ref{lemma:classicality equals homomorphism} gives us $A_{j_i}\in\la_A$ and $B_{k_i}\in\la_B$ such that $\p_i^* \subset A_{j_i}$ and $\p_i^*\subset B_{k_i}$. But then $\p_i^*\subset A_{j_i}\cap B_{k_i}$ for all $i$, so the by lemma \ref{lemma:classicality equals homomorphism} the partition $\la_{AB}=\{A_j\cap B_k\}$ is classical and is finer than both $\la_A$ and $\la_B$, contradicting our assumption. Therefore $\la_A=\la_B$, and the principle classical partition is unique.
\end{description}
\end{proof}

In what remains of this appendix we will enumerate the primitive preclusive multiplicative, $\M=\{\p_i\}_{i=1}^m$, where $m=|\M|$. If none of the duals of the co-events `overlap', so that $\p_i^*\cap\p_j^*=\emptyset$ for $i\neq j$, then it is easy to see (lemma \ref{lemma:principle classical partition construction} below) that the principle classical partition is given by :
$$\{\p_i^*|\p\in\M\}\sqcup\{\{\g\}|\g\not\in\p^*~for~any~\p\in\M\}.$$
However the duals of co-events in $\M$ will in general overlap, leading to a generalisation of the above. If $\la$ is a classical partition we already know (lemma \ref{lemma:classicality equals homomorphism}) that given $\p_i\in\M$ we can find $A_{m_i}\in\la$ such that $\p_i^*\subset A_{m_i}$. Now assume that $\p_i^*\cap\p_j^*\neq\emptyset$ for some $i\neq j$, then $\p_j^*\cap A_{m_i}\neq\emptyset$, so using lemma \ref{lemma:classicality equals homomorphism} $\p_j^*$ is also a subset of $A_{m_i}$. Repeating this argument, if we can find a $\p_k^*$ that has non-empty intersection with $\p_j^*$, then $\p_k^*\subset A_{m_i}$ even if $\p_k^*\cap\p_i^*=\emptyset$. Thus we really want to `group together' the $\p_i^*$ that are related by some sort of transitive closure of intersection; we achieve this by use of an equivalence relation.

First we define a relation, $\sim_\cap$, on the set $\Mstar\H :=\{\p^*|\p\in\M\}$ as follows:
$$\p_i^*\sim_\cap\p_j^* \Leftrightarrow \p_i^*\cap\p_j^*\neq\emptyset.$$
This relation is reflexive and symmetric, thus its transitive closure, $\approx_\cap$, is an equivalence relation, which we will call \emph{intersection equivalence}. We can therefore partition $\Mstar\H$ using $\approx_\cap$ to give us the set $\Mint\H$, and will denote by $\Phi^*(\p_i^*)$ the equivalence class containing $\p_i^*$. Finally, we can `merge' the elements of the equivalence class $\Phi^*(\p_i^*)$ to give us the set:
\beq
\PF_i:=\bigcup_{\p^*\in\Phi^*}\p^*.
\eeq
It is easy to see that $\PF_i\in\EA$ (since is it the union of elements of $\EA$), so its dual $\Phi^{F}_i:=\Phi^{F**}_i$ is a co-event. This leads us to define:
\bea
\MF\H := \{\PF_i | \Phi^*\in\Mint\H\}, \nonumber \\
{\cal{M}}_F\H := \{ \Phi^{F}_i | \Phi^*\in\Mint\H\}.
\eea
The $\PF_i$ (or their duals) are sometimes referred to as \emph{fat co-events}. The fat co-event corresponding to $\p_i^*$ is denoted by $\PF(\p_i^*)$. It is clear that the $\PF_i$ are disjoint, however they do not necessarily form a cover of $\O$. Thus if we define $Z$ to be the set of singleton sets of histories that are not an element of any $\PF_i$ (or equivalently any $\p_i^*$) we can form the partition:
\beq
\la_C = \MF\H \sqcup Z,
\eeq
of $\O$. It is easy to see that $\la_C$ is the principle classical partition with respect to $\mu$.

\begin{lemma}\label{lemma:principle classical partition construction}
Let $\H$ be a histories theory with a finite sample space. Then $\la_C$ as defined above is the principle classical partition with respect to $\mu$.
\end{lemma}
\begin{proof}
Every $\p\in\M$ is an element of an equivalence class $\Phi^*(\p)\in\Mint\H$, and thus is a subset of the fat co-event $\PF(\p)\in\la_C$. Then by lemma \ref{lemma:classicality equals homomorphism} $\la_C$ is a classical partition with respect to $\mu$.

Now let $\la = \{A_i\}$ be a classical partition (with respect to $\mu$). In the above we showed that if $\p_i^*\subset A_m$ then $\p_j^*\subset A_m$ whenever $\p_i^*\cap\p_j^*\neq\emptyset$. Then using an inductive argument it is easy to see that $\p_k^*\subset A_m$ whenever $\p_k^*\approx_\cap \p_i^*$, so that $\PF(\p_i^*)\subset A_m$. Further, for any element $\{\g\}$ of $Z$ we can find a $A_m$ such that $\{\g\}\subset A_m$, thus $\la$ can be considered a partition of $\la_C$. But since $\la$ was arbitrary, $\la_C$ is less than or equal to every classical partition; and so must be the principle classical partition.
\end{proof}

The principle classical partition helps us to further develop the notion of classicality for the multiplicative scheme, and more generally the histories approach. Note the contrast with the idea of classicality in consistent histories (see for example \cite{Griffiths:1996}), in which the absence in general of a `principle decoherent partition' allows the possibility of incompatible but `equally valid' classical interpretations (or `frameworks' \cite{Griffiths:1996}) of a given system.

Further, note that there is no way to distinguish between two intersectionally equivalent co-events using the classical partitions. Thus if we make the assumption that `observable events' are elements of classical partitions, this suggests that no sequence of measurements could allow us to distinguish between the distinct realities described by two co-events within the same intersectional equivalence class. In this way the fat co-events may be a useful practical rule, representing the `finest grained' information we could possibly discover using experimental techniques; a state of affairs we might think of as conforming to a \emph{co-event uncertainty principle}.

\bibliography{Bib}
\bibliographystyle{plain}


\end{document}